\documentclass{article}

\usepackage{iclr2025_conference,times}
\iclrfinalcopy

\usepackage{amsmath, bm, amssymb, amsthm}
\usepackage{multirow}
\usepackage{mathtools}
\usepackage{threeparttable}
\usepackage{tablefootnote}
\usepackage{xspace}
\usepackage{graphicx}
\usepackage{subfigure}

\usepackage[utf8]{inputenc} % allow utf-8 input
\usepackage[T1]{fontenc}    % use 8-bit T1 fonts
\usepackage{hyperref}       % hyperlinks
\usepackage{cleveref}
\usepackage{url}            % simple URL typesetting
\usepackage{booktabs}       % professional-quality tables
\usepackage{amsfonts}       % blackboard math symbols
\usepackage{nicefrac}       % compact symbols for 1/2, etc.
\usepackage{microtype}      % microtypography
\usepackage{algorithm}
\usepackage{algorithmic}
\usepackage{wrapfig,tikz}
\usepackage{chemformula}
\usepackage{adjustbox}
 \usepackage{tabularray}
\usepackage{subcaption}
\usepackage{diagbox}
\usepackage{booktabs}

\usepackage{amsmath,amsfonts,bm}

\def\eqref#1{equation~\ref{#1}}
\def\1{\bm{1}}

\DeclareMathAlphabet{\mathsfit}{\encodingdefault}{\sfdefault}{m}{sl}
\SetMathAlphabet{\mathsfit}{bold}{\encodingdefault}{\sfdefault}{bx}{n}

\title{Design of Ligand-Binding Proteins with Atomic Flow Matching}

\newcommand{\method}{\textsc{AtomFlow}\xspace}
\newcommand{\rfdiffusionaa}{\text{RFDiffusionAA}\xspace}

\newcommand{\rosettafoldaa}{\text{RoseTTAFoldAA}\xspace}
\newcommand{\rfdiffusion}{\text{RFDiffusion}\xspace}
\newtheorem{proposition}{Proposition}

\author{%
  Junqi Liu$^{1,3}$,
  Shaoning Li$^{2,3}$,
  Chence Shi$^{3,4,5}$,
  Zhi Yang$^{1,}$\footnotemark[2]~,
  Jian Tang$^{4,6,7, 3,}$\footnotemark[2] \\
  $^1$ School of Computer Science, Peking University $^2$ The Chinese University of Hong Kong \\
  $^3$ BioGeometry
  $^4$ Mila - Qu\'ebec AI Institute
  $^5$ Universit\'e de Montr\'eal 
  $^6$ HEC Montr\'eal \\
  $^7$ CIFAR AI Research Chair \\
  \texttt{\{liujunqi,yangzhi\}@pku.edu.cn}\\
  \texttt{jian.tang@hec.ca} 
}

\begin{document}
\sloppy

\maketitle

\renewcommand{\thefootnote}{}

\footnotetext[1]{
$^\dagger$~Corresponding authors.
}
\footnotetext[2]{
Work in progress, conducted during the internships of the first two authors at BioGeometry.
}
\renewcommand{\thefootnote}{\arabic{footnote}}
\vspace{-10pt}
\begin{abstract}
\vspace{-5pt}
Designing novel proteins that bind to small molecules is a long-standing challenge in computational biology, with applications in developing catalysts, biosensors, and more.
Current computational methods rely on the assumption that the binding pose of the target molecule is known, which is not always feasible, as conformations of novel targets are often unknown and tend to change upon binding.
In this work, we formulate proteins and molecules as unified biotokens, and present \method, a novel deep generative model under the flow-matching framework for the design of ligand-binding proteins from the 2D target molecular graph alone.
Operating on representative atoms of biotokens, \method captures the flexibility of ligands and generates ligand conformations and protein backbone structures iteratively.
We consider the multi-scale nature of biotokens and demonstrate that \method can be effectively trained on a subset of structures from the Protein Data Bank, by matching flow vector field using an SE(3) equivariant structure prediction network.
Experimental results show that our method can generate high fidelity ligand-binding proteins and achieve performance comparable to the state-of-the-art model \rfdiffusionaa, while not requiring bound ligand structures.
As a general framework, \method holds the potential to be applied to various biomolecule generation tasks in the future.

\end{abstract}

\vspace{-8pt}
\section{Introduction}
\label{sec:intro}

\vspace{-5pt}
Proteins are indispensable macromolecules that drive the essential processes of living organisms.
A crucial mechanism by which they accomplish this is through binding with small molecules~\citep{schreierComputationalDesignLigand2009}.
Continuous progress has been made to design ligand-binding proteins with various biological functions, such as catalysts and biosensors~\citep{bennettImprovingNovoProtein2023}.
However, the problem remains challenging due to the complex interactions between proteins and molecules, as well as the inherent flexibility of ligands.
The most well-established approaches depend on shape complementarity to dock molecules onto native protein scaffold structures~\citep{bick2017computational, polizzi2020defined}, which are computationally expensive.

Recently, \rfdiffusionaa~\citep{krishna2024generalized}, a de novo protein design method based on the all-atom structure prediction model \rosettafoldaa~\citep{krishna2024generalized}, has shown remarkable performance in designing novel ligand-binding proteins for small molecules.
This method explicitly captures the interactions between proteins and molecules, achieving superior performance compared to its predecessor \rfdiffusion~\citep{watson2023novo}, which can only model interactions between amino acid residues.
Despite their great potential for ligand-binding protein design, current approaches assume that the bound conformation of the target molecule is known and rigid.
However, the binding pose of the target molecule is not always available, especially for molecules that do not bind to any known natural proteins~\citep{bick2017computational}.
While it is possible to mitigate this limitation by sampling a diverse set of conformers and subsequently filtering them using expert knowledge~\citep{krishna2024generalized}, this approach demands potentially prohibitive computational resources.
Additionally, the constraint of ligand rigidity is suboptimal, as ligands often undergo significant conformation changes upon binding with proteins~\citep{mobley2009binding}.
We illustrate this phenomenon in Figure.\ref{fig:ligand_deform}.
Some pioneering efforts have been made to account for ligand flexibility~\citep{zhang2024pocketgen, stark2024harmonic}, however, these methods can only design the portions of proteins that directly interact with the ligands and require the rest part of the proteins as input.

\begin{wrapfigure}{r}{.5\textwidth}
    \vspace{-11pt}
    \begin{flushright}
    \includegraphics[width=0.5\textwidth]{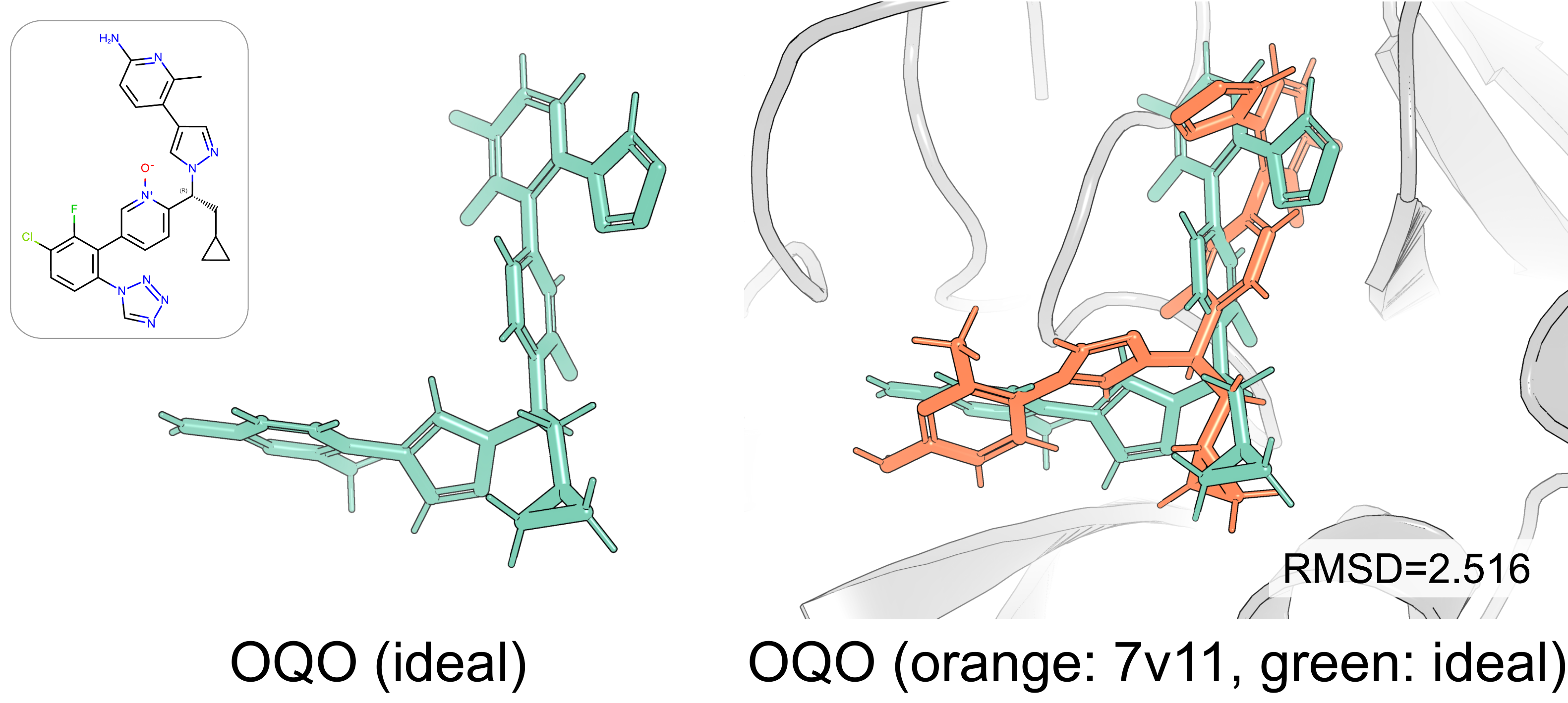}
    \end{flushright}
    \vspace{-8pt}
    \caption{The conformer of OQO deforms upon binding to coagulation factor XIa. Green: ideal conformer. Orange: bound conformer.}
    \label{fig:ligand_deform}
    \vspace{-13pt}
\end{wrapfigure}

To address the aforementioned issues, we present \underline{Atom}ic \underline{Flow}-matching (\method),
a novel deep generative model grounded in the
flow-matching framework~\citep{lipman2022flow} for the design of ligand-binding proteins from 2D molecular graphs alone.
\method considers ligand flexibility and can iteratively generate ligand conformations and bound protein backbone structures. Instead of relying on a fixed ligand conformer, \method learns to update the ligand structure along with the structure of the protein binder.
Inspired by recent advances in all-atom structure modeling~\citep{krishna2024generalized, abramson2024accurate}, we conceptualize proteins and molecules as biotokens with representative atoms, which are associated with various type-specific attributes and can be modeled by a single, unified network.
This approach maximizes the information aggregation between different molecular types, while encouraging the model to focus on the key interaction patterns.
Following the rectified flow approach~\citep{liu2022flow} for generative modelling, we define a flow on the representative atoms as a linear interpolation between the bound protein-ligand complex structures and noisy structures.
We demonstrate that, with minor approximations, the vector field of the defined flow can be effectively learned using an SE(3)-equivariant structure prediction module and a variant of Frame Aligned Point Error (FAPE) loss~\citep{jumper2021highly} that compensates for the multi-scale nature of their geometric features\footnote{The size of a protein is often much larger than that of a molecule. The size disparity should be considered when designing flow-matching models for stable training and inference.}.
After training, protein-ligand complex structures can be sampled from the approximated vector field, which iteratively transforms and refines noisy structures based on 2D molecular graphs. 
The idea of regressing the vector field using a structure prediction module is also explored in a concurrent work~\citep{jing2024alphafold}, but their focus is on protein structure prediction.
Notably, as a general generative model operating on biotokens, \method is versatile for different molecular types and has the potential to be applied to various biomolecule generation tasks.

We follow the \emph{in silico} evaluation pipeline of the state-of-the-art method \rfdiffusionaa, evaluating \method on several key metrics including self-consistency, binding affinity, diversity and novelty. \method matches the overall performance of \rfdiffusionaa and demonstrates advantages in various situations. An ablation study further highlights that when the bound structure is unknown, \method successfully designs protein binders with high binding affinity, whereas \rfdiffusionaa can be constrained by its dependence on a fixed, suboptimal ligand structure.
\vspace{-5pt}
\section{Related Work}
\label{sec:related}
\vspace{-3pt}

\textbf{Ligand-binding Protein Design.}
Traditional approaches to ligand-binding protein design mainly rely on docking molecules onto large sets of shape-complementary protein pockets~\citep{polizzi2020defined, lu2024novo}. While the screening process can be accelerated with deep learning models~\citep{an2023denovo}, conventional methods are computationally expensive and often depend on domain experts~\citep{bick2017computational}.
Recent advances in deep generative models have paved the way for data-driven approaches, and a variety of models have been proposed to design proteins conditioned on binding targets~\citep{shi2022protein, kong2023end, watson2023novo, zhang2024pocketgen}.
Focusing on molecule binder design, \rfdiffusion~\citep{watson2023novo} generates novel proteins from scratch, using a heuristic attractive-repulsive potential to measure shape complementarity.
The follow-up work \rfdiffusionaa~\citep{krishna2024generalized} improves the performance by explicitly modeling the interactions between proteins and molecules with an all-atom formulation.
These approaches assume binding poses of ligands are known and impose rigidity constraints on ligand structures.
Another line of research focuses on designing binding pockets for small molecules~\citep{stark2024harmonic, zhang2024pocketgen}. While taking ligand flexibility into consideration, they can only design the portions of proteins that interact with the ligands and require the rest part of the proteins as input.
Our model also accounts for the ligand flexibility, but is able to design full ligand-binding proteins from 2D molecular graph alone.

\textbf{Protein Generative Model and Structure Prediction.}
Recently, various deep generative models for protein generation have emerged~\citep{ingraham2023illuminating, lin2023generating, yim2023se, yim2023fast, wu2024protein, watson2023novo, krishna2024generalized}.
For example, Genie~\citep{lin2023generating} introduces a diffusion process defined on C$\alpha$ coordinates of proteins and allows for the incorporation of motif structures as conditions.
FrameDiff~\citep{yim2023se} takes a step further by generating novel protein backbone structures using an SE(3) diffusion process applied to residue frames.
Its successor, FrameFlow~\citep{yim2023fast}, accelerates the generation process by leveraging the flow-matching framework.
However, these approaches are tailored for single-chain protein generation and fall short in modeling multiple biomolecules.
In contrast, we treat multiple biomolecules, e.g., proteins and molecules, as biotokens and define a novel flow-matching model on their representative atoms.
This allows us to design ligand-binding proteins based solely on molecular graphs, effectively capturing the flexibility of biomolecules and the intricate interactions between them.
Our work is also related to approaches that perform protein structure predictions within the all-atom framework, such as \rosettafoldaa~\citep{krishna2024generalized} and AlphaFold 3~\citep{abramson2024accurate}.
These methods tokenize various types of biomolecules into unified tokens, aiming to develop a universal structure prediction model for all molecular types presented in the Protein Data Bank.
Our \method adopts the same practice, and we believe this formulation can maximize the information flow between proteins and molecules~\citep{bryant2024structure}, while our structural modelling on the representative atoms encourages the model to focus on the key patterns of biointeractions.
\vspace{-5pt}
\section{Preliminaries}
\label{sec:preliminary}
\vspace{-5pt}
\subsection{Notations and Problem Formulation}

\textbf{Notations.}~In this work, a protein-ligand complex is represented as a series of $N$ biotokens
$\{ a_i \mid a_i = (s_i, x_i), i = 1, 2, \dots, N \}$, where each token $a_i$ corresponds 
to either a protein residue or a ligand atom, $s_i$ denotes the token type, and $x_i \in \mathbb{R}^{3}$ 
denotes the token position, i.e. the coordinate of its representative atom. Let $\mathcal{S}_\text{protein}$ and $\mathcal{S}_\text{atom}$ be the set of amino acid types and chemical elements, respectively. For protein residues, $s_i \in \mathcal{S}_{\text{protein}}$, with $x_i$ being the position of the C-$\alpha$ carbon. 
 For ligand atoms, $s_i \in \mathcal{S}_{\text{atom}}$, 
with $x_i$ being the atomic position. We define the protein token set as 
$\mathcal{P} = \{a_i \mid s_i \in \mathcal{S}_{\text{protein}}\}$, with $N_p = |\mathcal{P}|$ 
being the number of protein residues, and the ligand token set as 
$\mathcal{M} = \{a_i \mid s_i \in \mathcal{S}_{\text{atom}}\}$, with $N_m = |\mathcal{M}|$ 
representing the number of ligand atoms. In our settings, $N=N_p+N_m$. The biotokens are attributed with token-level features $f^\text{token} \in \mathbb{R}^{N \times c_t}$ and pair-level features $f^\text{pair} \in \mathbb{R}^{N \times N \times c_p}$, where $c_t$ and $c_p$ denote the feature dimensions.

\textbf{Problem Formulation.}~Given a ligand molecule represented as a chemical graph $\mathcal{G} = (\mathcal{V}, \mathcal{E})$ and a residue count $N_p$ for the protein binder to be designed, we aim to generate a protein-ligand complex, where a conformer of $\mathcal{G}$ is docked to a protein binder with $N_p$ residues. Specifically, by describing the target protein-ligand complex as a series of biotokens, we generate the token positions $\{ x_i \}$, with $\mathbf{x}_m=\{ x_i \mid a_i \in \mathcal{M} \}$ being a valid conformer for $\mathcal{G}$, and $\mathbf{x}_p=\{ x_i \mid a_i \in \mathcal{P} \}$ being a protein binder with high binding affinity to $\mathbf{x}_m$. Following previous works \citep{krishna2024generalized, yim2023se}, we additionally generate the token frames $\{ T_i=(r_i, t_i) \mid a_i \in \mathcal{P}\}$ for protein tokens as described in Appendix~\ref{appx:frame}, which can be used to recover full backbone coordinates of residues. The design of residue types $\{ s_i \mid a_i \in \mathcal{P} \}$ is delegated to an existing reverse folding model~\citep{dauparas2023atomic}.

\vspace{-5pt}
\subsection{Flow Matching}

Building upon the significant success of diffusion models in various generative tasks, flow matching models~\citep{albergo2022building, liu2022flow} allows for faster and more reliable sampling from a distribution learnt from data. The generative process of flow matching models is usually defined by a probability path $p_t({x}), t\in [0, 1]$ that gradually transforms from a known noisy distribution $p_0({x})=q(x)$, such as $\mathcal{N}({x}|0, I)$ for ${x}\in\mathbb{R}$, to an approximate data distribution $p_1\approx p_\mathrm{data}(x)$. A vector field $u_t({x})$, which leads to an ODE $\frac{\mathrm{d}\phi_t(\mathbf{x})}{\mathrm{d}t} = u_t(\phi_t(\mathbf{x}))$, is used to generate the probability path via the push-forward equation,
\begin{equation}
    p_t = [\phi_t]_* p_0 = p_0(\phi_t^{-1}(x))\mathrm{det}\left[\frac{\partial \phi_t^{-1}}{\partial x}(x)\right],
\end{equation}
which could be approximated with a trainable network $\hat{v}_t(x;\theta)$.

Due to the complexity of defining an appropriate $p_t$ and $u_t$ , we could alternatively define a conditional probability path $p_t(x|x_1)$, which is usually derived through a conditional vector field $u_t(x|x_1)$ for each data point $x_1$~\citep{lipman2022flow}. The conditional vector field is then approximated with a trainable network $\hat{v}_t(x;\theta)$. \citet{lipman2022flow} has proved that the conditional flow matching loss,

\begin{equation}
\label{equation:cfm}
    \mathcal{L}_{\mathrm{CFM}}(\theta) = \mathbb{E}_{t, p_{\text{data}}(x_1), p_t(x|x_1)}\|\hat{v}_t(x;\theta) - u_t(x|x_1)\|,
\end{equation}

has identical gradients w.r.t. $\theta$ with $\mathcal{L}_\mathrm{FM}=\mathbb{E}_{t, p_\mathrm{data}(x)}||\hat{v}_t(x;\theta)-u_t(x)||$, which means the model can generate a marginal vector field by simply learning from the $x_1$-conditioned vector fields, without access to $p_t(x)$ and $u_t(x)$. After training, a neural ODE is obtained, ready for sampling from $p_0$ to $p_t$ by an ODE solver~\citep{jardine2011euler} .
\vspace{-5pt}
\section{Method}
\label{sec:method}

\method adopts a unified biotoken representation to generate the protein binder and ligand structure by learning the joint distribution of the token positions conditioned on a ligand chemical graph $\mathcal{G}$, $p(\{ x_i \}|\mathcal{G})$, from known structures of proteins and protein-ligand complexes. To achieve this, we define a rectified flow on the space of all token positions $\mathbf{x}\in \mathbb{R}^{N\times 3}$, and the corresponding vector field is approximated with an SE(3)-equivariant structure prediction module. The structure predicted at the last generation step is adopted as the final result. In this section, we introduce the flow matching model in Section 4.1, the biotoken feature representation in Section 4.2, the structure prediction module in Section 4.3, and the training and inference procedures in Section 4.4. The overview of our method is illustrated in Figure~\ref{fig:overview}.

\begin{figure}[ht]
\makebox[\textwidth][c] {
\includegraphics[width=\textwidth]{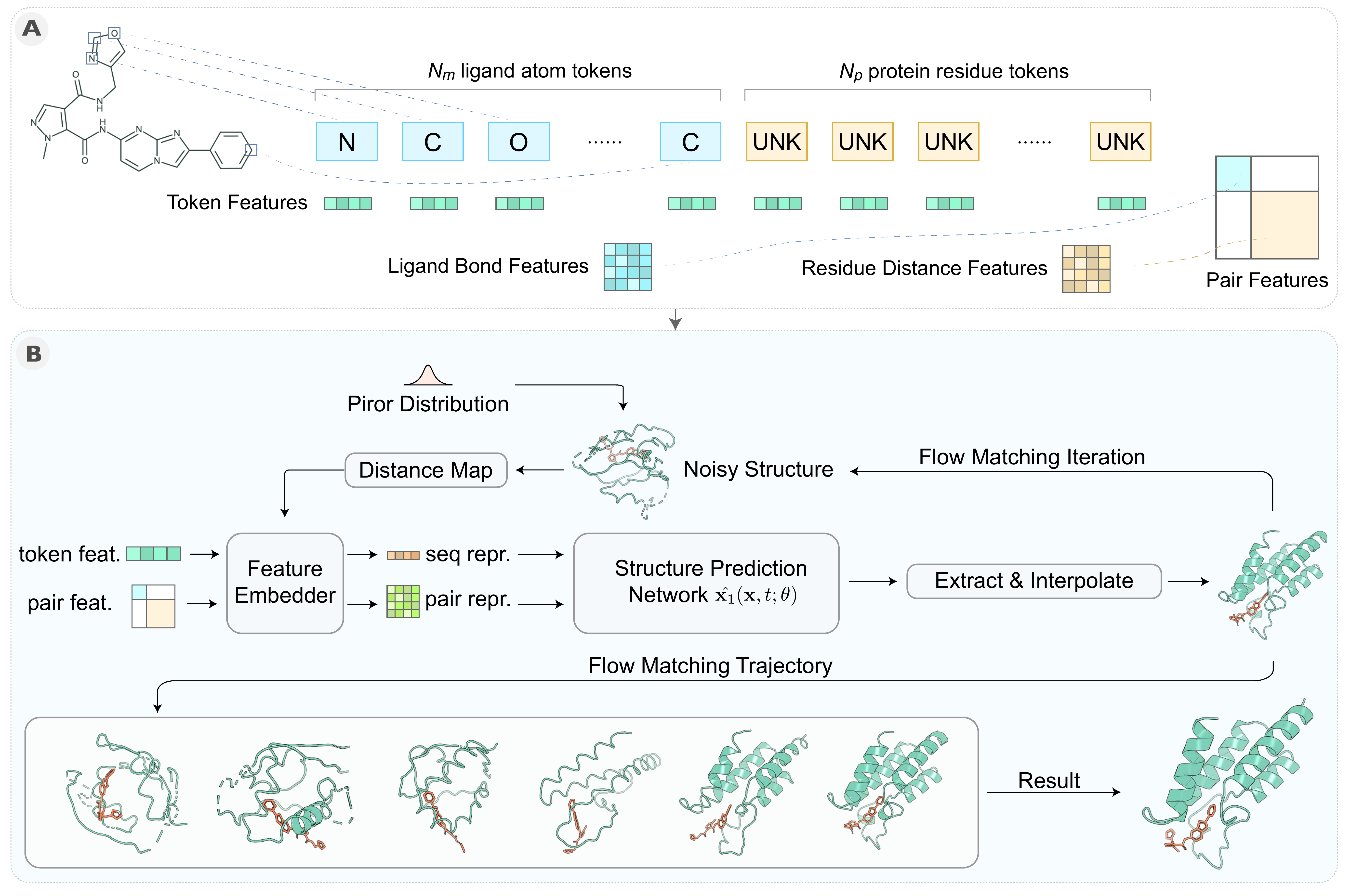}
}
\centering
\caption{The inference process of \method. We represent the protein-ligand complex as a series of biotokens and embed their token and pair level features. Starting from a noisy sample, the flow matching procedure gradually generates the designed structure $x_1$ with a structure prediction network. }
\label{fig:overview}
\end{figure}
\vspace{-10pt}

\subsection{Flow Matching for Protein-Ligand Complex Generation}

We jointly design the complex structure $\mathbf{x}=\mathbf{x}_m \cup \mathbf{x}_p$, which lies in the space of $\mathbb{R}^{N\times 3}$, with a flow matching model. Considering that different structures obtained under arbitrary SE(3) transformations correspond to the same complex, we treat each structure as an element in the quotient space $\mathcal{Q}: \mathbb{R}^{N \times 3} / \text{SE}(3)$, where two structures are identical if they could be perfectly aligned with an SE(3) transformation~\citep{jing2024alphafold}. This quotient space is proved to be a Riemannian manifold when defined with suitable care~\citep{diepeveen2024riemannian}. 

Following Riemannian Flow Matching~\citep{chen2024flow}, we define a rectified flow on this manifold with a premetric $d: \mathcal{Q}\times \mathcal{Q}\rightarrow \mathbb{R}$. We denote $\text{align}_x(y)$ for $x,y \in \mathbb{R}^{N\times 3}$ as aligning structure $y$ to $x$ to minimize RMSD, then the premetric $d(x,y)$ could be defined as the minimum point-wise root mean square deviation (RMSD) among all pairs of possible structures in the original space $\mathbb{R}^{N\times 3}$ for two elements in the quotient space
\begin{equation}
\label{equation:premetric}
    % d(x,y) = \|\text{align}_x(y)-x\|.
    d(x,y) = \min_{\tau \in \mathrm{SE}(3)} \text{RMSD}\,(\tau(y), x) = \text{RMSD}\,( \text{align}_x(y) - x )
\end{equation}
\begin{proposition}
\label{prop:premetric}
The premetric in equation \ref{equation:premetric} is a qualified premetric on $Q$.
\end{proposition}
With such premetric at hand, we could obtain a well-defined conditional vector field that decreases the premetric linearly from the prior distribution to the data distribution
\begin{equation}
\label{equation:cvf}
    u_t(x|x_1)=\dfrac{1}{1-t}\left(\text{align}_x(x_1)-x\right).
\end{equation}
We leave the proof of Proposition~\ref{prop:premetric} and the derivation of equation \ref{equation:cvf} to Appendix~\ref{appx:flow}.  Since the vector field is defined as a function of $\mathbf{x}_1$, we could learn the vector field with a structure prediction model $\hat{\mathbf{x}}_1(\mathbf{x},t;\theta)$. By substituting equations \ref{equation:cvf} into equation \ref{equation:cfm}, we obtain the training loss
\begin{equation}
\label{equation:cfm.final}
    \mathcal{L}_{\mathrm{CFM}}(\theta) = \mathbb{E}_{t, p_{\text{data}}(\mathbf{x}_1), p_t(\mathbf{x}|\mathbf{x}_1)}\left\|\dfrac{1}{1-t}(\text{align}_{\mathbf{x}}(\hat{\mathbf{x}}_1(x,t;\theta))-\text{align}_\mathbf{x}(\mathbf{x}_1))\right\|,
\end{equation}
This loss calculates the $(1-t)$-normalized distance between predicted $\hat{\mathbf{x}}_1$ and $\mathbf{x}_1$ in the data distribution aligned to the noisy structure of current step, which is SE(3)-equivariant to both the predicted and ground truth structure. The structure module is designed to predict the token frames (Section~\ref{sec:preliminary}), while the token positions are extracted from them during the generation process. The last prediction output is adopted as the final result.

Defining a unified flow matching procedure on the joint distribution enables the model to directly learn the structure characteristics that leads to a tightly binded complex, as well as the conformation deformation of both the proteins and the ligands, which is essential to designing a satisfactory ligand-binding protein. 

\vspace{-5pt}
\subsection{Representation of Conditional Features}
\label{subsec:representation}

The generation process of \method is conditioned on the ligand chemical graph $\mathcal{G}$ and a designated protein length $N_p$. We model such conditions as an additional condition to the vector field $u$. As a result, the inputs of the prediction network $\hat{\mathbf{x}}_1$ is augmented to accept conditional features. With the biotoken representation, we embed all such features as $f^\text{token}$ and $f^\text{pair}$ as illustrated in Figure~\ref{fig:overview}A.

For a ligand chemical graph $\mathcal{G}$, we embed the chemical element, as well as other known chemical properties as $f^\text{token}$ of ligand tokens. The chemical bonds $\mathcal{E}$ are embeded in $f^\text{pair}$ as a multi-dimensional adjacency tensor, each dimension representing a bond type. We also embed the relative residue position~\citep{shaw2018selfattention} as a pair feature, while the residue tokens may also be attributed with other known conditions. We concatenate the protein and ligand features to form a unified feature tensor, eliminating the need to distinguish different types of token when processing the features.

\vspace{-5pt}
\subsection{Structure Prediction Network}

The structure prediction network $\hat{\mathbf{x}}_1(\mathbf{x},t;\theta)$ \footnote{ Though $\hat{\mathbf{x}}_1$ is a function of $\mathbf{x},t,f^\text{token},f^{feat}$, we omitted certain parameters to simplify the text. } predicts the token frames $\{ T_i \}$, which can be used to extract token positions $\mathbf{x}_1$, given a series of noisy positions $\mathbf{x}$ at timestamp $t$. It encodes $\mathbf{x}$, along with $f^\text{token}$ and $f^\text{pair}$, with an SE(3) invariant encoding module, processing the representation with a transformer stack, and generates the predicted structure with a structure module based on invariant-point attention (IPA)~\citep{jumper2021highly}, as illustrated in Figure~\ref{fig:overview}B. The network jointly processes two kinds of biotokens, protein residues and ligand atoms, with different spatial scales, and handles such difference with special care.

\textbf{Distance Map.}~The input coordinates $\mathbf{x}$ are encoded by projecting the one-hot binned distance map between input coordinates for each token pair to the feature space
\begin{equation}
    t_{i,j}=\text{Linear}(\text{BinRepr}(\|\mathbf{x}^{(i)}-\mathbf{x}^{(j)}\|)),
\end{equation}
where the bins are not divided equally considering different precision requirement between residues and atoms. This representation is SE(3) invariant, since the internal distance does not change under rigid transformation. \footnote{ To accommodate the precision differences between ligands and proteins, the bin intervals are dense between 1Å (approximate length of a chemical bond) and 3.25Å (approximate distance between adjacent amino acids) and sparser beyond 3.25Å. }

\textbf{Feature Embedder.}~The feature embedder generates a single representation $s\in{\mathbb{R}}^{N\times c_s}$ and pair representation $z\in{\mathbb{R}}^{N\times N\times c_z}$ from distance map $h$, noise level $t$, $f^\text{token}$ and $f^\text{pair}$ for the following steps. The noise level is encoded with Gaussian Fourier embedding~\citep{song2021scorebased}. The local features are concatenated and projected to single representation $s$ and pair representation $z$, $s_i = \text{Linear}(f^{\text{local}}_i)$.
The pair features and input encoding are projected and added to $z$
\begin{equation}
    z_{i,j} = \text{Linear}(f^{\text{local}}_i) + \text{Linear}(f^{\text{local}}_j) + f^{\text{pair}}_{i,j} + t_{i,j}.
\end{equation}
As described in Section~\ref{subsec:representation}, different token types can be treated the same and processed uniformly.

\textbf{Structure Module.}~The structure module generates a predicted complex structure, represented as a series of token frames $T^{N}$. For ligand atoms, the rotation of the predicted frame is always identity rotation, while the translation equals to its position. It first processes $z$ through a deep transformer stack (Appendix~\ref{appx:model}) to obtain a denoised pair representation $z'$, and converts $s$ and $z'$ to $T^N$ through a series of shared-weight IPA block
\begin{equation}
    T_{1\cdots N} = \text{IPAStack}(s_{1\cdots N}, \text{TransformerStack}(z_{1\cdots N, 1\cdots N})).
\end{equation}
The IPA stack outputs a sequence of frames for each token, while the rotations for atom tokens are dropped and replaced with the atom frame demonstrated in Section~\ref{subsec:representation}. The final output represents the full complex structure $\hat{\mathbf{x}}$, while token positions $\hat{x_1}$ is calculated as previously described. The Transformer stack on the unified token sequence allows us to smoothly model the interactions between different types of biological entities in a joint feature space, while the IPA blocks are proved to be efficient when the final structure is properly embedded in the transformer output~\citep{jumper2021highly}.

\textbf{Auxiliary Head.}~We add an auxiliary head to predict the pairwise binned distance from the denoised pair representation $z'$, $h_i=\text{softmax}(\text{Linear}(z'_i))$, which directly supervise the input of structure module and has been proved to be helpful during training~\citep{jumper2021highly}. The bins are also unevenly divided to accommodate the multi-scale characteristics of the predicted complex.

\vspace{-5pt}
\subsection{Training and Inference}

We train the network $\hat{\mathbf{x}}_1$ by sampling data points and timestamps, calculating the noisy input, and supervising the predicted results. At inference time, we transforms the token positions sampled from the prior distribution through the predicted vector field with an ODE solver, and outputs the structure we obtained at the final step.

\textbf{Loss.}~We supervise the predicted complex structure $T$ with a metric that measures the structural difference between the observed structure and the predicted structure. Preliminary experiments show that the $\mathcal{L}_\text{CFM}$ in equation~\ref{equation:cfm.final} leads to a fluctuating training trajectory since the aligning object $\mathbf{x}$ varies upon training. With approximation (Appendix~\ref{appx:model}), we replace the loss function to a variant of the widely-adopted FAPE function~\citep{jumper2021highly},

\begin{equation}
\label{equation:cfm.fape}
    \mathcal{L}_{\text{CFM-FAPE}}(\theta) = \mathbb{E}_{t, p_{\text{data}}(\mathbf{x}_1), p_t(\mathbf{x}|\mathbf{x}_1)}\left[\dfrac{1}{1-t}\mathrm{FAPE}(\hat{\mathbf{x}}_1(\mathbf{x},t;\theta), \mathbf{x}_1)\right].
\end{equation}
We show that this substitution does not change the training objective in the appendix. Since the normalization factor $Z$ in the FAPE loss is related to the numerical range of distance, we divide the FAPE loss into protein-protein interaction, protein-ligand interaction, ligand-ligand interaction, and assign different $Z$s for the three parts. For the auxiliary head, we adopt the cross-entropy loss averaged over all token pairs for the predicted distance. The final training loss 
\begin{equation}
\mathcal{L}=\alpha_1\mathcal{L}_\text{CFM-FAPE-pp}+\alpha_2\mathcal{L}_\text{CFM-FAPE-pl}+\alpha_3\mathcal{L}_\text{CFM-FAPE-ll}+\alpha_4\mathcal{L}_\text{aux}.
\end{equation}
Further details are elaborated in Appendix~\ref{appx:model}.

\textbf{Training.}~We sample the timestamp $t$ from the logit normal distribution, assigning more weight on intermediate steps, which helps the model to achieve better performance on hard timestamps~\citep{esser2024scaling, karras2022elucidating}. The prior distribution $q(x)$ is selected as $\mathcal{N}(0, \sigma_\text{data})$, where $\sigma_\text{data}=10$. The input $\mathbf{x}$ is given by interpolating the data point and a sample from the prior distribution. The training procedure is shown in Algorithm~\ref{algo:train}.

\vspace{-10pt}
\begin{figure}[h]
\centering
\begin{minipage}{0.48\textwidth}
\begin{algorithm}[H]
\caption{Training}
\label{algo:train}
\begin{algorithmic}[1]
    \REQUIRE data distribution $p(\mathbf{x})$, prior distribution $q(\mathbf{x})$, trainable model parameters $\theta$
    \WHILE{not converged}
        \STATE sample complex structure $\mathbf{x}_1$ and its corresponding ligand chemical graph $\mathcal{G}$ from $p(\mathbf{x}), t \sim [0,1), \mathbf{x}_0 \sim q(\mathbf{x})$
        \STATE $N, f^\text{token}, f^\text{pair} \leftarrow \text{Embedder}(\mathcal{G}, N_p)$
        \STATE $\mathbf{x}_t \leftarrow t\cdot \mathbf{x}_1 + (1-t)\cdot \text{align}_{\mathbf{x}}(\mathbf{x}_0)$
        \STATE $\theta \leftarrow \text{Optimizer}(\theta, (\mathbf{x}_t, f^\text{token}, f^\text{pair}, t), \mathcal{L})$
    \ENDWHILE
    \RETURN $\theta$
\end{algorithmic}
\end{algorithm}
\end{minipage}
\hfill
\begin{minipage}{0.48\textwidth}
\begin{algorithm}[H]
\caption{Inference}
\label{algo:infer}
\begin{algorithmic}[1]
    \REQUIRE Chemical graph $\mathcal{G}$, residue count $N_p$, scheduler $t_{0\cdots m}$, prior distribution $q(\mathbf{x})$, model parameters $\theta$
    \STATE $N, f^\text{token}, f^\text{pair} \leftarrow \text{Embedder}(\mathcal{G}, N_p)$
    \STATE sample token positions $\mathbf{x}_{t_0}\sim q(\mathbf{x})$
    \FOR{$i=0$ to $m-1$}
        \STATE $T_{1\cdots N} \leftarrow \hat{\mathbf{x}}_1(\mathbf{x}_{t_i}, f^\text{token}, f^\text{pair}, t_i;\theta)$ 
        \STATE $\hat{\mathbf{x}}_1 \leftarrow \text{Extract}(T)$
        \STATE calculate $\mathbf{x}_{t_i+1}$ as Equation~\ref{equation:infer}
    \ENDFOR
    \RETURN $T$
\end{algorithmic}
\end{algorithm}
\end{minipage}

\end{figure}

\textbf{Inference.}~A scheduler of noise levels $\{ t_i \}_{i=0}^{m},  t_0=0, t_m=1$ is used to determine the noise level $t_i$ of each sampling step $x_{t_i}$. Starting from a noisy sample $x_{t_i}=x_0$ as the initial model input, the structure prediction network predicts the vector field, which gives $x_{t_{i+1}}$ with the Euler's Method, i.e. 
\begin{equation}
\label{equation:infer}
    \mathbf{x}_{t_{i+1}} = \mathbf{x}_{t_i} + \dfrac{t_{i+1}-t_i}{1-t_i}\bigg(\text{align}_{\mathbf{x}_{t_i}}\Big(\text{Extract}\big(\hat{\mathbf{x}}_1(\mathbf{x}_{t_i},t_i;\theta)\big)\Big)-\mathbf{x}_{t_i}\bigg),
\end{equation}
where the Extract function extracts the token positions from the predicted token frames. The model output at the last step is adopted as the final result. The inference procedure is shown in Algorithm~\ref{algo:infer}.

\vspace{-5pt}
\section{Experiments}
\label{sec:experiment}

Following previous protein design models~\citep{yim2023fast,lin2023generating,watson2023novo} and binder design models~\citep{krishna2024generalized}, we evaluate \method through in silico experiments on key metrics of our generated binder including self-consistency, binding affinity, diversity and novelty.
\vspace{-5pt}
\subsection{Experiment Setup}

\textbf{Training Data.}~We train the denoising model on two datasets: PDBBind~\citep{liu2017forging}, a protein-ligand conformer dataset derived from the Protetin Data Bank (PDB)~\citep{berman2000protein}, and SCOPe~\citep{chandonia2022scope}, a structure categorical dataset for protein. The model is firsted trained on solely generating the protein structure for 40k steps, and then finetuned on co-generating both of the protein and ligand structure for 30k steps. 

\textbf{Baseline and Model Variant.}~We compare \method with the state-of-the-art binder generation method \rfdiffusionaa~\citep{krishna2024generalized}, which is extensive trained on almost all known data. Since \rfdiffusionaa requires a fixed ligand structure at the binding state as input, we extend our method to work under its setting. For \method, besides the original setting (\method-N), we also train a version of our model with the pairwise distance matrix of the bound structure as an auxiliary hint input (\method-H). This version still needs to generate the ligand structure itself, rather than rely on a fixed structure, as other specifications is not modified.

\textbf{Evaluation Set.}~We mainly evaluate all methods on a selected ligand set (evaluation set) from \rfdiffusionaa (FAD, SAM, IAI, OQO). The evaluation set comprises ligands from inside and outside the training set, with both long and short lengths. We conduct evaluation on an extended ligand set (extended set, see Appendix~\ref{appx:eval}) to further demonstrate the performance of \method.
\vspace{-5pt}
\subsection{Self-consistency and Conformer Legitimacy}
\label{sec:experiment:sc}

In this section, we evaluate the legitimacy of the generated protein structure by self-consistency RMSD and the predicted ligand structure at binding state by detecting structural violence in the conformer. Legitimacy is crucial in binder design, given that the model output is not guaranteed to be valid, while a design with higher legitimacy is more possible to fold as expected.

\begin{figure}[htbp]
\includegraphics[width=\textwidth]{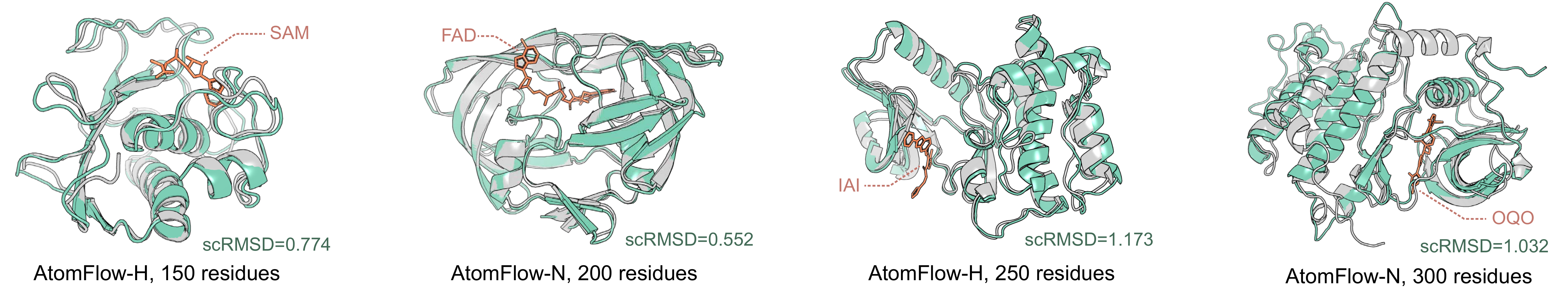}
\caption{Designed structures for different ligands at different lengths. We align the ESMFold predicted structure to the designed structure, and report the scRMSD metric. Green: designed protein; Orange: designed ligand conformer; Grey: ESMFold predicted protein. }
\label{fig:design_sample}
\end{figure}

\textbf{Protein Structure.} For protein structures, self-consistency RMSD is widely adopted as a metric to evaluate their legitimacy~\citep{lin2023generating, watson2023novo}, which compares the generated structure and the folding of its sequence predicted by an accurate model. We adopt LigandMPNN~\citep{dauparas2023atomic} to predict possible sequences from the generated structures. We first generate 8 sequences for all designed structure with LigandMPNN, then predict the corresponding protein structure with ESMFold~\citep{lin2023evolutionary}, and the metric for each generated structure is calculated as the minimum rooted mean squared distance between the designed structure and predicted structure (scRMSD). For each ligand in the evaluation set, we generates 10 structures for lengths in [100, 150, 200, 250, 300]. The results are shown in Table~\ref{table:sc}. We illustrate several generated samples in Figure~\ref{fig:design_sample}, and the cumulative distribution of scRMSD among them in Figure~\ref{fig:sc-and-vina}A and \ref{fig:sc-and-vina}D. The results on the extended set are shown in Appendix~\ref{appx:eval}.

\begin{table}[h]
\begin{center}
\begin{tabular}{cccccc}
\hline
    Method & Overall & SAM & FAD & IAI & OQO \\
    \hline
    \method-H & \textbf{0.57} & \textbf{0.60} & 0.36 & \textbf{0.58} & \textbf{0.74} \\
    \method-N & 0.50 & 0.50 & 0.38 & \textbf{0.58} & \textbf{0.54} \\
    RFDiffusionAA & 0.52 & \textbf{0.60} & \textbf{0.58} & 0.48 & 0.42 \\
    RFDiffusion & 0.33 & 0.04 & 0.50 & 0.44 & 0.32 \\
\hline
\end{tabular}
\end{center}
\vspace{-5pt}
\caption{Proportion of samples with scRMSD $< 2$ on the evaluation set (higher is better).}
\label{table:sc}
\end{table}
\vspace{-5pt}

\method and \rfdiffusionaa outperforms \rfdiffusion on all ligands in the evaluation set, while both \method-H and \method-N reach comparable results to \rfdiffusionaa, and  exhibits advantages over \rfdiffusionaa on several cases. The restricted performance of \rfdiffusion is as expected since its binding potential for guiding the protein-ligand interaction may lead to structural destruction. Both \method and \rfdiffusionaa models the interaction directly, thus not requiring a strong potential to interfere the generative process, and lead to better generation results. Notably, without relying on structural guidance from the input ligand conformer, \method-N achieves close performance to \method-H, thereby successfully augmenting the setting to flexible design.

\vspace{-5pt}
\subsection{Binding Affinity}
\label{sec:experiment:bf}

\begin{figure}[t]
\begin{center}
\includegraphics[width=1.0\textwidth]{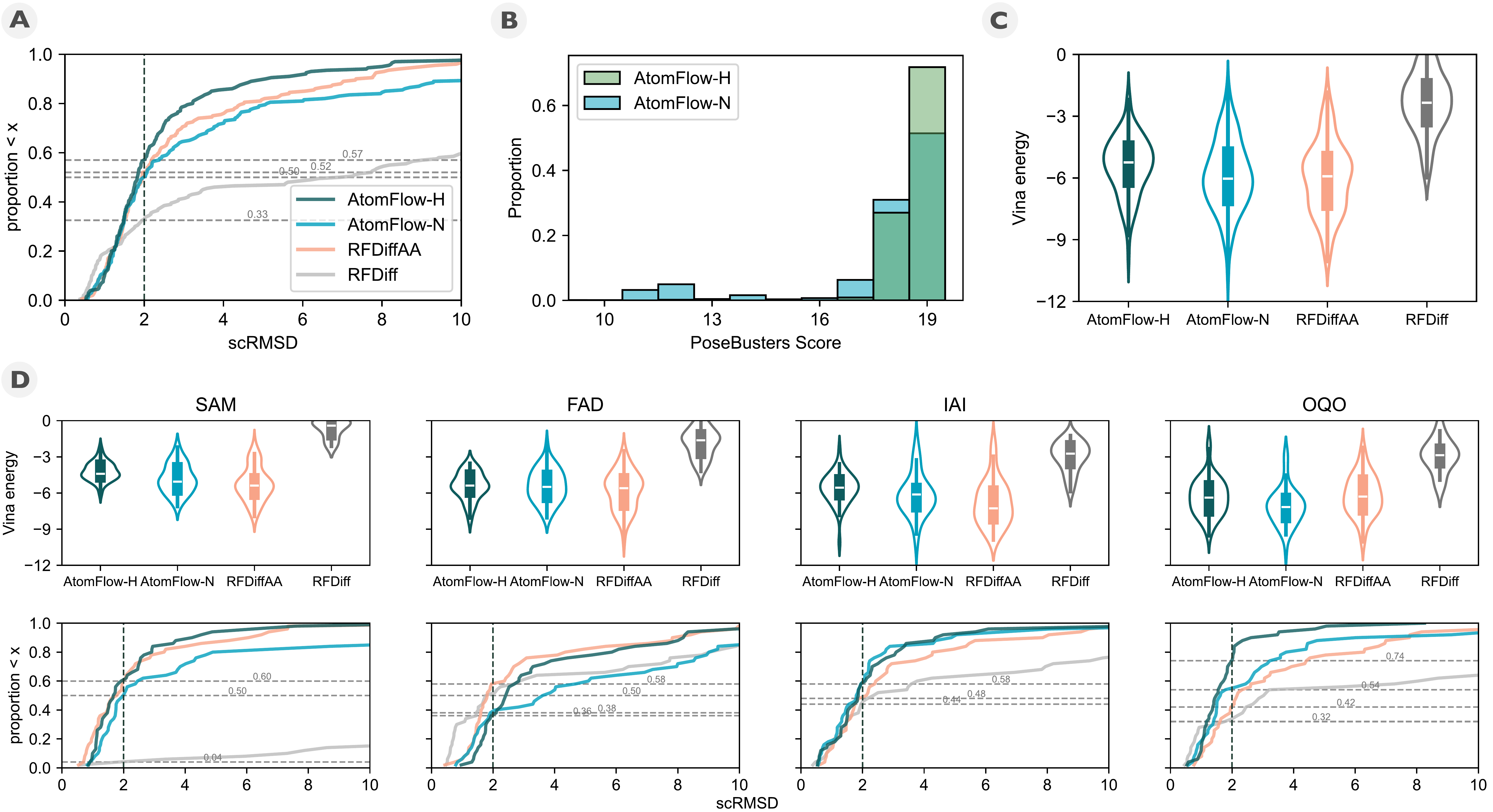}
\end{center}
\vspace{-5pt}
\caption{\textbf{A}: Self-consistency RMSD distribution curve demonstrating the ratio among all designed samples for the evaluation set with $\text{scRMSD} \le x$ (higher is better). \method outperforms \rfdiffusion with a curve similar to \rfdiffusionaa. \method-H generates achieves the best result among the methods. The ratio of samples with $\text{scRMSD}<2$ is highlighted. \textbf{B}: PoseBusters score distribution of \method generated samples on the extended set. Most ligand conformers generated by \method-N only fails $\le 1$ metric of its evaluations. \textbf{C}: Vina score distribution over all designs on the evaluation set (lower is better). \method achieves comparable performance to \rfdiffusionaa, outperforming \rfdiffusion. \textbf{D}: scRMSD curve and Vina energy distribution over designs for each ligand in the evaluation set. \method outperforms \rfdiffusion on all cases and metrics. \method and \rfdiffusionaa each exhibit advantages on different ligands, with comparable overall results.}
\vspace{-5pt}
\label{fig:sc-and-vina}
\end{figure}

In this section, we evaluate the binding affinity of the designed protein binder by calculating an energy function for the atom-level interaction between the protein and the ligand. Binding affinity is the key metric to reveal whether the designed binders are able to bind the target molecule. Though the real binding affinity could only be determined through experiments in the wet lab, an energy function is usually adopted as an in silico alternative~\citep{zhang2024pocketgen}. We calculate the AutoDock Vina Score~\citep{eberhardt2021autodock} for all 8 sequences packed by the Rosetta packer~\citep{leaver2011rosetta3}, and the reported energy for a structure is the minimum score among all packed proteins. We calculate the energy for all generated structures for the selected ligand set in Section~\ref{sec:experiment:sc} and compare \method with \rfdiffusionaa and \rfdiffusion. The result is illustrated in Figure~\ref{fig:sc-and-vina}C and \ref{fig:sc-and-vina}D. 

We find that the binding affinity of \rfdiffusion is quite poor since it does not model the protein-ligand interaction directly. \method has reached comparable binding affinity to \rfdiffusionaa, though marginally lower on several cases. We attribute this to the exhaustive training process of \rfdiffusionaa on all known data in PDB, while \method could be further trained and this will be investigated further in our future work. The minimum binding energy generated by \method-H is slightly higher than that of \method-N, possibly because the provided conformer hint hinders the model from exploring additional binding states.

\begin{figure}[htbp]
\begin{center}
\includegraphics[width=\textwidth]{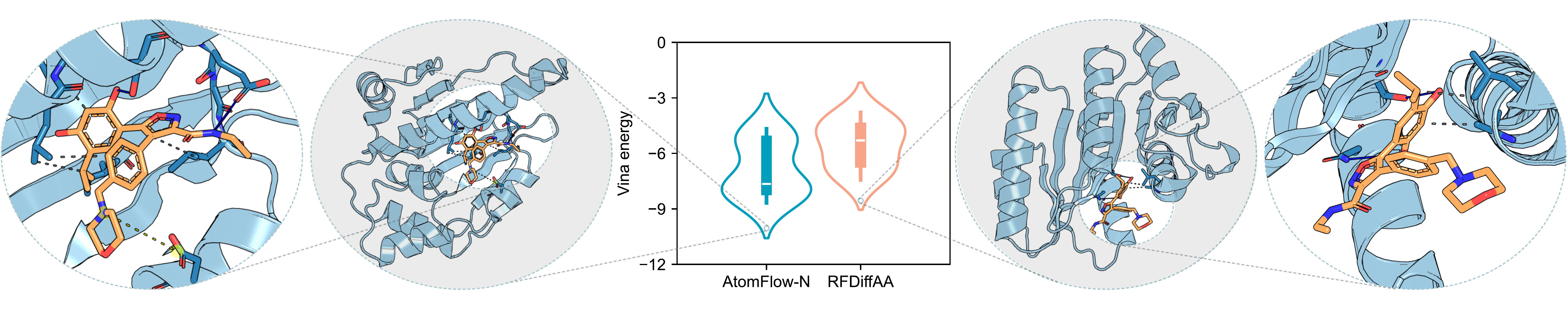}
\end{center}
\caption{\method-N designs binders with lower vina energy distribution than \rfdiffusionaa on 2GJ without the bound structure. Illustration of one sample for each method with PLIP demonstrates that the \method-N designed binder has more chemical interaction with the ligand. }
\label{fig:vina-abl}
\vspace{-5pt}
\end{figure}

We further compare \method with \rfdiffusionaa on a realistic setting where the bound conformer is unknown. We set the target ligand as luminespib (PDB id: 2GJ), an Hsp90 inhibitor~\citep{piotrowska2018activity}. A designed protein binder for luminespib may act as a protein drug carrier to enhance drug efficacy. Luminespib is a molecule ligand with 33 heavy atoms, so that the conformer is quite flexible when docked to different receptors. We design 10 binders for luminespib using \method and \rfdiffusionaa. The ideal conformer from PDB is provided to \rfdiffusionaa, while no conformer is provided to \method. The binding energy of the designed structures and one designed sample with PLIP~\citep{adasme2021plip} to demonstrate the protein-ligand interaction are illustrated in Figure~\ref{fig:vina-abl}. It is shown that \method generates more binders with higher binding affinity than \rfdiffusionaa, and significantly outperforms \rfdiffusionaa on the lowest energy among all generated structures. This demonstrate that a proper bound structure is crucial to the performance of \rfdiffusionaa, while \method does not rely on such structure and generates proper conformers by co-modelling the structure space of proteins and ligands.
\vspace{-5pt}
\subsection{Diversity and Novelty}
\label{sec:experiment:dv}

In this section, we report the diversity and novelty of \method, following common practice in literature~\citep{krishna2024generalized, yim2023se}. Diversity refers to the structural divergence of the designed binders for a certain ligand, while novelty refers to how close a designed protein is to the known proteins. For diversity, we generate 100 structures with 200 residues for each ligand, and then use MaxCluster~\citep{maxcluster} to calculate pairwise structural distance of the outputs and report the number of clusters using different thresholds of maximum distance within cluster. For novelty, we generate 4 structures with residue count in $[100, 101, \cdots, 300]$ for each ligand, and then calculate the highest TM-score~\citep{zhang2005tmalign} between a designed structure and any similar structure searched by FoldSeek~\citep{kempen2024fast} (pdbTM), as well as the protein scRMSD. The search range of pdbTM is all known protein structures in PDB. 

\begin{figure}[htb]
\centering
\includegraphics[width=0.8\textwidth]{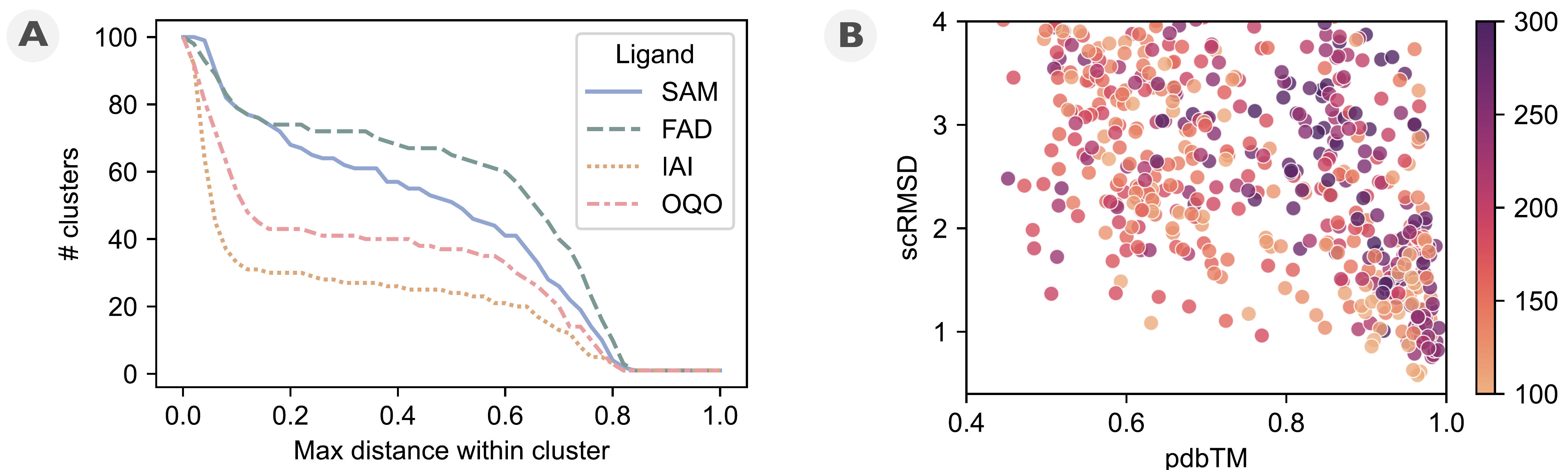}
\caption{\textbf{A}: Cluster count based on different thresholds for maximum difference within cluster for each ligand in the evaluation set. \method generates diverse binder folds for all ligands, not restricted to the existing binder structure. \textbf{B}: Scatter plot of designability (scRMSD) vs. novelty (pdbTM) for ligands in the evaluation set. \method successfully designs self-consistent structures with high pdbTM, demonstrating high novelty.}
\label{fig:cluster-pdbtm}
\end{figure}

\vspace{-10pt}
Figure~\ref{fig:cluster-pdbtm}A shows that the structures generated by \method is quite diverse for all four ligands, and the diversity varies among different ligands. Though existing protein-ligand complexes only provides limited folds for possible binders, by adding protein-only data to the training set, our model successfully learns from the protein structure distribution to generate more possible folds, instead of replicating known patterns. The scatter plot of scRMSD vs. pdbTM shown in Figure~\ref{fig:cluster-pdbtm}B reveals that \method has the ability to generate structures that are quite different from existing proteins with acceptable designability. Note that most designable structures are still similar to known ones, which is as expected since most protein folds are already discovered, while novel folds are quite sparse and hard to derive.

\vspace{-5pt}
\section{Conclusion and Future Work}
\label{sec:conclusion}
\vspace{-5pt}

In this work, we proposed \method, a de novo protein binder design method for small molecule ligands considering the flexibility of ligand structure. Unlike previous works, \method no longer relies on a given bound ligand conformer as input. We represent the protein-ligand complex as unified biotokens, learning the structure distribution of both the proteins and the ligands simultaneously from the data with an SE(3)-equivariant flow matching model on the representative atoms. During evaluation, \method shows comparable design quality to the state-of-the-art model \rfdiffusionaa, which requires the ligand conformer to be fixed before design. Further evaluation exhibits the advantage of \method at the circumstance when the ligand conformer is not known. A direct future work is to support more precise control of the generated structures, and we're working to migrate \method to all kinds of biomolecules, including DNA, RNA, and metal ions.

\bibliography{reference.bib}

\begin{thebibliography}{47}
\providecommand{\natexlab}[1]{#1}
\providecommand{\url}[1]{\texttt{#1}}
\expandafter\ifx\csname urlstyle\endcsname\relax
  \providecommand{\doi}[1]{doi: #1}\else
  \providecommand{\doi}{doi: \begingroup \urlstyle{rm}\Url}\fi

\bibitem[Abramson et~al.(2024)Abramson, Adler, Dunger, Evans, Green, Pritzel, Ronneberger, Willmore, Ballard, Bambrick, et~al.]{abramson2024accurate}
Josh Abramson, Jonas Adler, Jack Dunger, Richard Evans, Tim Green, Alexander Pritzel, Olaf Ronneberger, Lindsay Willmore, Andrew~J Ballard, Joshua Bambrick, et~al.
\newblock Accurate structure prediction of biomolecular interactions with alphafold 3.
\newblock \emph{Nature}, pp.\  1--3, 2024.

\bibitem[Adasme et~al.(2021)Adasme, Linnemann, Bolz, Kaiser, Salentin, Haupt, and Schroeder]{adasme2021plip}
Melissa~F Adasme, Katja~L Linnemann, Sarah~Naomi Bolz, Florian Kaiser, Sebastian Salentin, V~Joachim Haupt, and Michael Schroeder.
\newblock Plip 2021: Expanding the scope of the protein--ligand interaction profiler to dna and rna.
\newblock \emph{Nucleic acids research}, 49\penalty0 (W1):\penalty0 W530--W534, 2021.

\bibitem[Ahdritz et~al.(2024)Ahdritz, Bouatta, Floristean, Kadyan, Xia, Gerecke, O’Donnell, Berenberg, Fisk, Zanichelli, et~al.]{ahdritz2024openfold}
Gustaf Ahdritz, Nazim Bouatta, Christina Floristean, Sachin Kadyan, Qinghui Xia, William Gerecke, Timothy~J O’Donnell, Daniel Berenberg, Ian Fisk, Niccol{\`o} Zanichelli, et~al.
\newblock Openfold: Retraining alphafold2 yields new insights into its learning mechanisms and capacity for generalization.
\newblock \emph{Nature Methods}, pp.\  1--11, 2024.

\bibitem[Albergo \& Vanden-Eijnden(2022)Albergo and Vanden-Eijnden]{albergo2022building}
Michael~S Albergo and Eric Vanden-Eijnden.
\newblock Building normalizing flows with stochastic interpolants.
\newblock \emph{arXiv preprint arXiv:2209.15571}, 2022.

\bibitem[An et~al.(2023)An, Said, Tran, Majumder, Goreshnik, Lee, Juergens, Dauparas, Anishchenko, Coventry, Bera, Kang, Levine, Alvarez, Pillai, Norn, Feldman, Zorine, Hicks, Li, Sanchez, Vafeados, Salveson, Vorobieva, and Baker]{an2023denovo}
Linna An, Meerit~Y Said, Long Tran, Sagardip Majumder, Inna Goreshnik, Gyu~Rie Lee, David Juergens, Justas Dauparas, Ivan~V. Anishchenko, Brian Coventry, Asim~K. Bera, Alex Kang, Paul~M. Levine, Valentina Alvarez, Arvind Pillai, Christoffer~H Norn, David Feldman, Dmitri Zorine, Derrick~R. Hicks, Xinting Li, Mariana~Garcia Sanchez, Dionne~K. Vafeados, Patrick~J. Salveson, Anastassia~A. Vorobieva, and David Baker.
\newblock De novo design of diverse small molecule binders and sensors using shape complementary pseudocycles.
\newblock \emph{bioRxiv}, 2023.
\newblock URL \url{https://api.semanticscholar.org/CorpusID:266540105}.

\bibitem[Bennett et~al.(2023)Bennett, Coventry, Goreshnik, Huang, Allen, Vafeados, Peng, Dauparas, Baek, Stewart, DiMaio, De~Munck, Savvides, and Baker]{bennettImprovingNovoProtein2023}
Nathaniel~R. Bennett, Brian Coventry, Inna Goreshnik, Buwei Huang, Aza Allen, Dionne Vafeados, Ying~Po Peng, Justas Dauparas, Minkyung Baek, Lance Stewart, Frank DiMaio, Steven De~Munck, Savvas~N. Savvides, and David Baker.
\newblock Improving de novo protein binder design with deep learning.
\newblock \emph{Nature Communications}, 14\penalty0 (1):\penalty0 2625, May 2023.
\newblock ISSN 2041-1723.
\newblock \doi{10.1038/s41467-023-38328-5}.

\bibitem[Berman et~al.(2000)Berman, Westbrook, Feng, Gilliland, Bhat, Weissig, Shindyalov, and Bourne]{berman2000protein}
Helen~M. Berman, John~D. Westbrook, Zukang Feng, Gary Gilliland, T.~N. Bhat, Helge Weissig, Ilya~N. Shindyalov, and Philip~E. Bourne.
\newblock The {Protein} {Data} {Bank}.
\newblock \emph{Nucleic Acids Research}, 28\penalty0 (1):\penalty0 235--242, 2000.

\bibitem[Bick et~al.(2017)Bick, Greisen, Morey, Antunes, La, Sankaran, Reymond, Johnsson, Medford, and Baker]{bick2017computational}
Matthew~J Bick, Per~J Greisen, Kevin~J Morey, Mauricio~S Antunes, David La, Banumathi Sankaran, Luc Reymond, Kai Johnsson, June~I Medford, and David Baker.
\newblock Computational design of environmental sensors for the potent opioid fentanyl.
\newblock \emph{Elife}, 6:\penalty0 e28909, 2017.

\bibitem[Bryant et~al.(2024)Bryant, Kelkar, Guljas, Clementi, and No{\'e}]{bryant2024structure}
Patrick Bryant, Atharva Kelkar, Andrea Guljas, Cecilia Clementi, and Frank No{\'e}.
\newblock Structure prediction of protein-ligand complexes from sequence information with umol.
\newblock \emph{Nature Communications}, 15\penalty0 (1):\penalty0 4536, 2024.

\bibitem[Chandonia et~al.(2022)Chandonia, Guan, Lin, Yu, Fox, and Brenner]{chandonia2022scope}
John-Marc Chandonia, Lindsey Guan, Shiangyi Lin, Changhua Yu, Naomi~K Fox, and Steven~E Brenner.
\newblock Scope: improvements to the structural classification of proteins -- extended database to facilitate variant interpretation and machine learning.
\newblock \emph{Nucleic Acids Research}, 50\penalty0 (D1):\penalty0 D553--D559, 2022.

\bibitem[Chen \& Lipman(2024)Chen and Lipman]{chen2024flow}
Ricky~TQ Chen and Yaron Lipman.
\newblock Flow matching on general geometries.
\newblock In \emph{The Twelfth International Conference on Learning Representations}, 2024.

\bibitem[Dauparas et~al.(2023)Dauparas, Lee, Pecoraro, An, Anishchenko, Glasscock, and Baker]{dauparas2023atomic}
Justas Dauparas, Gyu~Rie Lee, Robert Pecoraro, Linna An, Ivan Anishchenko, Cameron Glasscock, and David Baker.
\newblock Atomic context-conditioned protein sequence design using ligandmpnn.
\newblock \emph{Biorxiv}, pp.\  2023--12, 2023.

\bibitem[Diepeveen et~al.(2024)Diepeveen, Esteve-Yag{\"u}e, Lellmann, {\"O}ktem, and Sch{\"o}nlieb]{diepeveen2024riemannian}
Willem Diepeveen, Carlos Esteve-Yag{\"u}e, Jan Lellmann, Ozan {\"O}ktem, and Carola-Bibiane Sch{\"o}nlieb.
\newblock Riemannian geometry for efficient analysis of protein dynamics data.
\newblock \emph{Proceedings of the National Academy of Sciences}, 121\penalty0 (33):\penalty0 e2318951121, 2024.

\bibitem[Eberhardt et~al.(2021)Eberhardt, Santos-Martins, Tillack, and Forli]{eberhardt2021autodock}
Jerome Eberhardt, Diogo Santos-Martins, Andreas~F. Tillack, and Stefano Forli.
\newblock Autodock {Vina} 1.2.0: New {Docking} {Methods}, {Expanded} {Force} {Field}, and {Python} {Bindings}.
\newblock \emph{Journal of Chemical Information and Modeling}, 61\penalty0 (8):\penalty0 3891--3898, 2021.

\bibitem[Esser et~al.(2024)Esser, Kulal, Blattmann, Entezari, M{\" u}ller, Saini, Levi, Lorenz, Sauer, Boesel, Podell, Dockhorn, English, and Rombach]{esser2024scaling}
Patrick Esser, Sumith Kulal, Andreas Blattmann, Rahim Entezari, Jonas M{\" u}ller, Harry Saini, Yam Levi, Dominik Lorenz, Axel Sauer, Frederic Boesel, Dustin Podell, Tim Dockhorn, Zion English, and Robin Rombach.
\newblock Scaling {Rectified} {Flow} {Transformers} for {High}-{Resolution} {Image} {Synthesis}.
\newblock In \emph{Forty-first {International} {Conference} on {Machine} {Learning}}, volume abs/2403.03206, 2024.

\bibitem[Herbert(2008)]{maxcluster}
Alex Herbert.
\newblock Maxcluster: A tool for protein structure comparison and clustering, 2008.
\newblock URL \url{http://www.sbg.bio.ic.ac.uk/~maxcluster/}.

\bibitem[Ingraham et~al.(2023)Ingraham, Baranov, Costello, Barber, Wang, Ismail, Frappier, Lord, Ng-Thow-Hing, Van~Vlack, et~al.]{ingraham2023illuminating}
John~B Ingraham, Max Baranov, Zak Costello, Karl~W Barber, Wujie Wang, Ahmed Ismail, Vincent Frappier, Dana~M Lord, Christopher Ng-Thow-Hing, Erik~R Van~Vlack, et~al.
\newblock Illuminating protein space with a programmable generative model.
\newblock \emph{Nature}, 623\penalty0 (7989):\penalty0 1070--1078, 2023.

\bibitem[Jardine(2011)]{jardine2011euler}
Dick Jardine.
\newblock Euler’s method in euler’s words.
\newblock \emph{Mathematical Time Capsules: Historical Modules for the Mathematics Classroom}, \penalty0 (77):\penalty0 215, 2011.

\bibitem[Jing et~al.(2024)Jing, Berger, and Jaakkola]{jing2024alphafold}
Bowen Jing, Bonnie Berger, and Tommi Jaakkola.
\newblock Alphafold meets flow matching for generating protein ensembles.
\newblock \emph{arXiv preprint arXiv:2402.04845}, 2024.

\bibitem[Jumper et~al.(2021)Jumper, Evans, Pritzel, Green, Figurnov, Ronneberger, Tunyasuvunakool, Bates, {\v{Z}}{\'\i}dek, Potapenko, et~al.]{jumper2021highly}
John Jumper, Richard Evans, Alexander Pritzel, Tim Green, Michael Figurnov, Olaf Ronneberger, Kathryn Tunyasuvunakool, Russ Bates, Augustin {\v{Z}}{\'\i}dek, Anna Potapenko, et~al.
\newblock Highly accurate protein structure prediction with alphafold.
\newblock \emph{nature}, 596\penalty0 (7873):\penalty0 583--589, 2021.

\bibitem[Karras et~al.(2022)Karras, Aittala, Aila, and Laine]{karras2022elucidating}
Tero Karras, Miika Aittala, Timo Aila, and Samuli Laine.
\newblock Elucidating the design space of diffusion-based generative models.
\newblock \emph{Advances in neural information processing systems}, 35:\penalty0 26565--26577, 2022.

\bibitem[Kempen et~al.(2024)Kempen, Kim, Tumescheit, Mirdita, Lee, Gilchrist, S{\" o}ding, and Steinegger]{kempen2024fast}
Michel~van Kempen, Stephanie~S. Kim, Charlotte Tumescheit, Milot Mirdita, Jeongjae Lee, Cameron L.~M. Gilchrist, Johannes S{\" o}ding, and Martin Steinegger.
\newblock Fast and accurate protein structure search with {Foldseek}.
\newblock \emph{Nature Biotechnology}, 42\penalty0 (2):\penalty0 243--246, 2024.

\bibitem[Kingma(2014)]{kingma2014adam}
Diederik~P Kingma.
\newblock Adam: A method for stochastic optimization.
\newblock \emph{arXiv preprint arXiv:1412.6980}, 2014.

\bibitem[Kong et~al.(2023)Kong, Huang, and Liu]{kong2023end}
Xiangzhe Kong, Wenbing Huang, and Yang Liu.
\newblock End-to-end full-atom antibody design.
\newblock \emph{arXiv preprint arXiv:2302.00203}, 2023.

\bibitem[Krishna et~al.(2024)Krishna, Wang, Ahern, Sturmfels, Venkatesh, Kalvet, Lee, Morey-Burrows, Anishchenko, Humphreys, et~al.]{krishna2024generalized}
Rohith Krishna, Jue Wang, Woody Ahern, Pascal Sturmfels, Preetham Venkatesh, Indrek Kalvet, Gyu~Rie Lee, Felix~S Morey-Burrows, Ivan Anishchenko, Ian~R Humphreys, et~al.
\newblock Generalized biomolecular modeling and design with rosettafold all-atom.
\newblock \emph{Science}, 384\penalty0 (6693):\penalty0 eadl2528, 2024.

\bibitem[Leaver-Fay et~al.(2011)Leaver-Fay, Tyka, Lewis, Lange, Thompson, Jacak, Kaufman, Renfrew, Smith, Sheffler, et~al.]{leaver2011rosetta3}
Andrew Leaver-Fay, Michael Tyka, Steven~M Lewis, Oliver~F Lange, James Thompson, Ron Jacak, Kristian~W Kaufman, P~Douglas Renfrew, Colin~A Smith, Will Sheffler, et~al.
\newblock Rosetta3: an object-oriented software suite for the simulation and design of macromolecules.
\newblock In \emph{Methods in enzymology}, volume 487, pp.\  545--574. Elsevier, 2011.

\bibitem[Lin \& AlQuraishi(2023)Lin and AlQuraishi]{lin2023generating}
Yeqing Lin and Mohammed AlQuraishi.
\newblock Generating novel, designable, and diverse protein structures by equivariantly diffusing oriented residue clouds.
\newblock \emph{arXiv preprint arXiv:2301.12485}, 2023.

\bibitem[Lin et~al.(2023)Lin, Akin, Rao, Hie, Zhu, Lu, Smetanin, Verkuil, Kabeli, Shmueli, et~al.]{lin2023evolutionary}
Zeming Lin, Halil Akin, Roshan Rao, Brian Hie, Zhongkai Zhu, Wenting Lu, Nikita Smetanin, Robert Verkuil, Ori Kabeli, Yaniv Shmueli, et~al.
\newblock Evolutionary-scale prediction of atomic-level protein structure with a language model.
\newblock \emph{Science}, 379\penalty0 (6637):\penalty0 1123--1130, 2023.

\bibitem[Lipman et~al.(2022)Lipman, Chen, Ben-Hamu, Nickel, and Le]{lipman2022flow}
Yaron Lipman, Ricky~TQ Chen, Heli Ben-Hamu, Maximilian Nickel, and Matt Le.
\newblock Flow matching for generative modeling.
\newblock \emph{arXiv preprint arXiv:2210.02747}, 2022.

\bibitem[Liu et~al.(2022)Liu, Gong, and Liu]{liu2022flow}
Xingchao Liu, Chengyue Gong, and Qiang Liu.
\newblock Flow straight and fast: Learning to generate and transfer data with rectified flow.
\newblock \emph{arXiv preprint arXiv:2209.03003}, 2022.

\bibitem[Liu et~al.(2017)Liu, Su, Han, Liu, Yang, Li, and Wang]{liu2017forging}
Zhihai Liu, Minyi Su, Li~Han, Jie Liu, Qifan Yang, Yan Li, and Renxiao Wang.
\newblock Forging the {Basis} for {Developing} {Protein}--{Ligand} {Interaction} {Scoring} {Functions}.
\newblock \emph{Accounts of Chemical Research}, 50\penalty0 (2):\penalty0 302--309, 2017.

\bibitem[Lu et~al.(2024)Lu, Gou, Tan, Mann, Yang, Zhong, Gazgalis, Valdiviezo, Jo, Wu, et~al.]{lu2024novo}
Lei Lu, Xuxu Gou, Sophia~K Tan, Samuel~I Mann, Hyunjun Yang, Xiaofang Zhong, Dimitrios Gazgalis, Jes{\'u}s Valdiviezo, Hyunil Jo, Yibing Wu, et~al.
\newblock De novo design of drug-binding proteins with predictable binding energy and specificity.
\newblock \emph{Science}, 384\penalty0 (6691):\penalty0 106--112, 2024.

\bibitem[Mobley \& Dill(2009)Mobley and Dill]{mobley2009binding}
David~L Mobley and Ken~A Dill.
\newblock Binding of small-molecule ligands to proteins:“what you see” is not always “what you get”.
\newblock \emph{Structure}, 17\penalty0 (4):\penalty0 489--498, 2009.

\bibitem[Piotrowska et~al.(2018)Piotrowska, Costa, Oxnard, Huberman, Gainor, Lennes, Muzikansky, Shaw, Azzoli, Heist, et~al.]{piotrowska2018activity}
Z~Piotrowska, DB~Costa, GR~Oxnard, M~Huberman, JF~Gainor, IT~Lennes, A~Muzikansky, AT~Shaw, CG~Azzoli, RS~Heist, et~al.
\newblock Activity of the hsp90 inhibitor luminespib among non-small-cell lung cancers harboring egfr exon 20 insertions.
\newblock \emph{Annals of Oncology}, 29\penalty0 (10):\penalty0 2092--2097, 2018.

\bibitem[Polizzi \& DeGrado(2020)Polizzi and DeGrado]{polizzi2020defined}
Nicholas~F Polizzi and William~F DeGrado.
\newblock A defined structural unit enables de novo design of small-molecule--binding proteins.
\newblock \emph{Science}, 369\penalty0 (6508):\penalty0 1227--1233, 2020.

\bibitem[Schneider et~al.(2015)Schneider, Sayle, and Landrum]{schneider2015get}
Nadine Schneider, Roger~A Sayle, and Gregory~A Landrum.
\newblock Get your atoms in order—an open-source implementation of a novel and robust molecular canonicalization algorithm.
\newblock \emph{Journal of chemical information and modeling}, 55\penalty0 (10):\penalty0 2111--2120, 2015.

\bibitem[Schreier et~al.(2009)Schreier, Stumpp, Wiesner, and H{\"o}cker]{schreierComputationalDesignLigand2009}
Bettina Schreier, Christian Stumpp, Silke Wiesner, and Birte H{\"o}cker.
\newblock Computational design of ligand binding is not a solved problem.
\newblock \emph{Proceedings of the National Academy of Sciences}, 106\penalty0 (44):\penalty0 18491--18496, November 2009.
\newblock ISSN 0027-8424, 1091-6490.
\newblock \doi{10.1073/pnas.0907950106}.

\bibitem[Shaw et~al.(2018)Shaw, Uszkoreit, and Vaswani]{shaw2018selfattention}
Peter Shaw, Jakob Uszkoreit, and Ashish Vaswani.
\newblock Self-{Attention} with {Relative} {Position} {Representations}.
\newblock In \emph{North {American} {Chapter} of the {Association} for {Computational} {Linguistics} ({NAACL})}, pp.\  464--468, 2018.

\bibitem[Shi et~al.(2022)Shi, Wang, Lu, Zhong, and Tang]{shi2022protein}
Chence Shi, Chuanrui Wang, Jiarui Lu, Bozitao Zhong, and Jian Tang.
\newblock Protein sequence and structure co-design with equivariant translation.
\newblock \emph{arXiv preprint arXiv:2210.08761}, 2022.

\bibitem[Song et~al.(2021)Song, Sohl-Dickstein, Kingma, Kumar, Ermon, and Poole]{song2021scorebased}
Yang Song, Jascha Sohl-Dickstein, Diederik~P. Kingma, Abhishek Kumar, Stefano Ermon, and Ben Poole.
\newblock Score-{Based} {Generative} {Modeling} through {Stochastic} {Differential} {Equations}.
\newblock \emph{International Conference on Learning Representations (ICLR)}, 2021.

\bibitem[Stark et~al.(2024)Stark, Jing, Barzilay, and Jaakkola]{stark2024harmonic}
Hannes Stark, Bowen Jing, Regina Barzilay, and Tommi Jaakkola.
\newblock Harmonic {Self}-{Conditioned} {Flow} {Matching} for joint {Multi}-{Ligand} {Docking} and {Binding} {Site} {Design}.
\newblock In \emph{Forty-first {International} {Conference} on {Machine} {Learning}}, volume abs/2310.05764, 2024.

\bibitem[Watson et~al.(2023)Watson, Juergens, Bennett, Trippe, Yim, Eisenach, Ahern, Borst, Ragotte, Milles, Wicky, Hanikel, Pellock, Courbet, Sheffler, Wang, Venkatesh, Sappington, Torres, Lauko, Bortoli, Mathieu, Ovchinnikov, Barzilay, Jaakkola, DiMaio, Baek, and Baker]{watson2023novo}
Joseph~L. Watson, David Juergens, Nathaniel~R. Bennett, Brian~L. Trippe, Jason Yim, Helen~E. Eisenach, Woody Ahern, Andrew~J. Borst, Robert~J. Ragotte, Lukas~F. Milles, Basile I.~M. Wicky, Nikita Hanikel, Samuel~J. Pellock, Alexis Courbet, William Sheffler, Jue Wang, Preetham Venkatesh, Isaac Sappington, Susana~V{\' a}zquez Torres, Anna Lauko, Valentin~De Bortoli, Emile Mathieu, Sergey Ovchinnikov, Regina Barzilay, Tommi~S. Jaakkola, Frank DiMaio, Minkyung Baek, and David Baker.
\newblock De novo design of protein structure and function with {RFdiffusion}.
\newblock \emph{Nature}, 620\penalty0 (7976):\penalty0 1089--1100, 2023.

\bibitem[Wu et~al.(2024)Wu, Yang, van~den Berg, Alamdari, Zou, Lu, and Amini]{wu2024protein}
Kevin~E Wu, Kevin~K Yang, Rianne van~den Berg, Sarah Alamdari, James~Y Zou, Alex~X Lu, and Ava~P Amini.
\newblock Protein structure generation via folding diffusion.
\newblock \emph{Nature communications}, 15\penalty0 (1):\penalty0 1059, 2024.

\bibitem[Yim et~al.(2023{\natexlab{a}})Yim, Campbell, Foong, Gastegger, Jim{\'e}nez-Luna, Lewis, Satorras, Veeling, Barzilay, Jaakkola, et~al.]{yim2023fast}
Jason Yim, Andrew Campbell, Andrew~YK Foong, Michael Gastegger, Jos{\'e} Jim{\'e}nez-Luna, Sarah Lewis, Victor~Garcia Satorras, Bastiaan~S Veeling, Regina Barzilay, Tommi Jaakkola, et~al.
\newblock Fast protein backbone generation with se (3) flow matching.
\newblock \emph{arXiv preprint arXiv:2310.05297}, 2023{\natexlab{a}}.

\bibitem[Yim et~al.(2023{\natexlab{b}})Yim, Trippe, De~Bortoli, Mathieu, Doucet, Barzilay, and Jaakkola]{yim2023se}
Jason Yim, Brian~L Trippe, Valentin De~Bortoli, Emile Mathieu, Arnaud Doucet, Regina Barzilay, and Tommi Jaakkola.
\newblock Se (3) diffusion model with application to protein backbone generation.
\newblock \emph{arXiv preprint arXiv:2302.02277}, 2023{\natexlab{b}}.

\bibitem[Zhang(2005)]{zhang2005tmalign}
Y.~Zhang.
\newblock Tm-align: a protein structure alignment algorithm based on the {TM}-score.
\newblock \emph{Nucleic Acids Research}, 33\penalty0 (7):\penalty0 2302--2309, 2005.

\bibitem[Zhang et~al.(2024)Zhang, Shen, Liu, and Zitnik]{zhang2024pocketgen}
Zaixi Zhang, Wanxiang Shen, Qi~Liu, and Marinka Zitnik.
\newblock Pocketgen: Generating {Full}-{Atom} {Ligand}-{Binding} {Protein} {Pockets}.
\newblock \emph{bioRxiv}, 2024.

\end{thebibliography}
\bibliographystyle{iclr2025_conference}

\clearpage
\appendix
\section{Appendix}
\subsection{Protein Frames}
\label{appx:frame}

Proteins are composed of amino acid chains linked by peptide bonds, forming a backbone with protruding side chains. Each amino acid’s position and orientation is described by a local coordinate system, or protein frame, centered on three key backbone atoms: the alpha carbon (C$\alpha$), the carbonyl carbon (C), and the amide nitrogen (N). These atoms act as reference points for establishing the frame. The alpha carbon (C$\alpha$) typically acts as the origin. The vector from C$\alpha$ to the amide nitrogen (N) is normalized to define one axis of the frame. A second axis is defined by the normalized vector from C$\alpha$ to the carbonyl carbon (C). The third axis is formed by the cross product of these two vectors, creating an orthogonal, right-handed coordinate system. The residue frame is typically represented as an SE(3) transformation \( T = (R, t) \), which maps a vector from this local system to the global coordinate system. In this transformation, \( t \) corresponds to the position of C$\alpha$ in the global system, and \( R \) represents the rotation needed to align the residue’s structure within the global context.

\begin{figure}[h]
\centering
\includegraphics[width=0.4\textwidth]{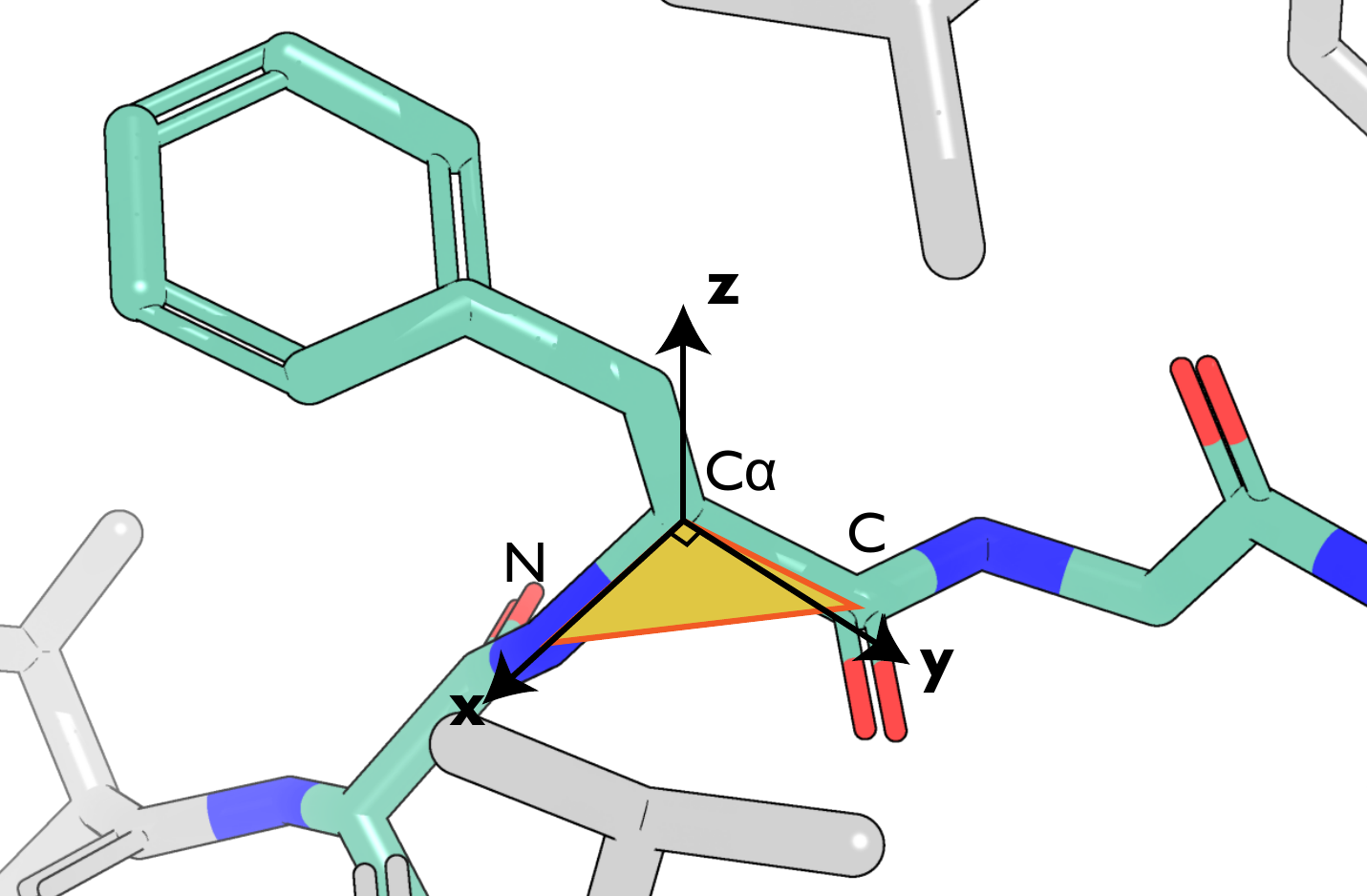}
\caption{A protein frame illustration. The C$\alpha$, C, N atoms form a panel, which is the xy panel. The x axis is defined as the orientation from C$\alpha$ to N, while the y axis is on the panel and perpendicular to the x axis. The z axis is perpendicular to the xy panel.}
\end{figure}

\subsection{Details on Biotokens}
\label{appx:token}

\paragraph{Token Features}

For ligand atom tokens, the token-level feature set includes: chirality, degree, formal charge, implicit valence, number of H atoms, number of radical electrons, orbital hybridisation, aromaticity, ring size. The pair-level feature is provided as one-hot embedding of the bond type. For residue tokens, no token-level feature is known, while the pair-level feature only contains the binned distance of residue index between residues. All features are encoded as a one-hot vector and concatenated.

\paragraph{Token Frames}

The final loss we adopted $\mathcal{L}_{\text{CFM-FAPE}}$ requires aligning the predicted structure to the local frame of every token. The frames of protein residues can be naturally defined as in Section \ref{sec:preliminary}. However, the frames of ligand atoms could not be choosed directly. Since a frame could be calculated from the coordinate of 3 atoms, we need to choose an atom triplet for every atom token. 

We first obtain an canonical rank of every atom that does not depend on the input order \citep{schneider2015get}. The atoms are then renamed to its rank. For atoms $ x $ with a degree greater than or equal to 2, we select the lexicographically smallest triplet $ (u, x, v) $ to define the frame, where $ u $ and $ v $ are neighbors of $ x $ . For atoms with a degree of 1, $ u $ is the only neighbor of $ x $ , and $ v $ is chosen as one of $ u $ 's neighbors. This method ensures that each atom's frame is defined in a consistent manner, irrespective of its position in the input sequence, thereby facilitating the model to learn a consistent structural target.

\paragraph{Extending Token Types and Features}

Though \method only considers the interaction between protein and molecule ligands, the unified biotoken has the potential to extend to all biological entities, including DNA, RNA, etc, by defining the token position, token frame, local and pair features, and the representation of the internal structure. For example, an RNA can be represented as a sequence of nucleotides, with the token position defined as its mass center, and the token frame calculated from an atom triplet, such as C2-N1-C6. 

The token features can also be extended to support more types of known information. For example, the local features could also contain an embedding to indicate preferred secondary-structure, or whether a ligand atom is required to be closer to the designed protein; the pair features could also contain the motif information with a distance map.

\subsection{Details on the Flow Matching Process}
\label{appx:flow}

For all types of tokens, we only consider their token positions to simplify the flow matching process. Thus, the positions of all tokens lies in the euclidean space $\mathbb{R}^{N\times 3}$. Since a complex could be arbitrarily moved or rotated in the coordinate space without changing its structure, we need an algorithm that treats different position series as the same if they could be aligned with an SE(3) translation. Thus, every data point we consider now lies in the quotient space $\mathbb{R}^{N\times 3}/\text{SE(3)}$. This quotient space is proved to be a Riemannian manifold~\citep{diepeveen2024riemannian}. 

For a Riemannian manifold, the flow matching process could be defined using a premetric \citep{chen2024flow}. A premetric $d: \mathcal{M}\times \mathcal{M}\rightarrow \mathbb{R}$ should satisfy: 1. $d(x,y)\ge 0\text{ for all }x,y\in \mathcal{M}$; 2. $d(x,y)=0\text{ iff }x=y$; 3. $\nabla d(x,y)\neq 0\text{ iff }x\neq y$.

We define our premetric as the minimum point-wise rooted sum of squared distance (RMSD) among all pairs of possible structures in the original space $\mathbb{R}^{N\times 3}$ for two elements in the quotient space \(
    d(x,y) = \|\text{align}_x(y)-x\| \), which satisfies all three conditions (Proposition~\ref{prop:premetric}).

\begin{proof}
Since the premetric is defined as a norm, it satisfies condition 1 by nature. When $x=y$, the best alignment that aligns $y$ to $x$ could derive the exact same position as $x$, yielding an zero norm. When $x\neq y$, when $y$ is aligned to $x$, there's still structural difference between the structures, thus the premetric is not zero. For condition 3, by defining $y'=\text{align}_x(y)$, we have 
\begin{equation}
\label{equation:nabla_d}
    \nabla d(x,y)=\nabla \sqrt{\sum_{i=1}^{n} (y_i'-x_i)^2} =\dfrac{y'-x}{||y'-x||}=\dfrac{\text{align}_x(y)-x}{||\text{align}_x(y)-x||}\ge 0.
\end{equation}
Thus $d(x,y)$ satisfies all the conditions as a qualified premetric.
\end{proof}

With such premetric, and a monotonically decreasing differentiable scheduler $\kappa(t)=1-t$, we could obtain a well-defined conditional vector field that linearly interpolates between the noisy and real data
\citep{chen2024flow}
\begin{equation}
\label{equation:cvf-apx}
    u_t(x|x_1)=\dfrac{\mathrm{d}\log \kappa(t)}{\mathrm{d}t}d(x,x_1) \dfrac{\nabla d(x,x_1)}{\|\nabla d(x,x_1)\|^2}=\dfrac{1}{1-t}(\text{align}_x(x_1)-x).
\end{equation}
The vector field in equation \ref{equation:cvf-apx} is calculated by substituting equation \ref{equation:nabla_d} into the left side.
This vector field provides the direction for moving straight towards $x_1$, and generates a probability flow that interpolate linearly between noisy sample $x_0$ and data sample $x_1$.

Since the vector field is defined as a function of $x_1$, we could learn the vector field with a structure prediction model $\hat{x_1}(x,t;\theta)$. By substituting equation \ref{equation:cvf} into equation \ref{equation:cfm}, we obtain the training loss
\begin{equation}
    \mathcal{L}_{\mathrm{CFM}}(\theta) = \mathbb{E}_{t, p_{\text{data}}(x_1), p_t(x|x_1)}\left\|\dfrac{1}{1-t}(\text{align}_{x}(\hat{x_1}(x,t;\theta))-\text{align}_x(x_1))\right\|.
\end{equation}

\subsection{Details on the Prediction Network}
\label{appx:model}

\paragraph{Structure Module Specifications}

The main components of the structure module is derived from Alphafold 2~\citep{jumper2021highly}, while our implementation builds on top of the widely acknowledged reimplementation OpenFold~\citep{ahdritz2024openfold}. The TransformerStack consists of 14 layers of simplified Evoformer block, and the IPAStack consists of 4 layers of Invariant Point Attention (IPA) blocks. The MSA operations in the Evoformer block is simplified by replacing the operations on the MSA feature matrix to the single representation $s_i$. The weights of the IPA blocks are shared, and the structural loss is calculated on the outputs of each block and averaged.

\paragraph{Training Details}
During training, we equally sample data from the SCOPe dataset (v2.08) and the PDBBind dataset (2020). We simply drop the data with more than 512 tokens, and we don't crop the filtered complexes since the cutoff is large enough and only filters out a relative small portion of data. We train our model on 10 NIVIDA RTX 4090 acceleration card, with a batch size set to 10, which means the batch size on each device is set to 1. We use the Adam Optimizer~\citep{kingma2014adam} with a weight-decaying learning rate scheduler, starting from $10^{-3}$ and decays the learnign rate by $0.95$ every 50k steps. We seperate the training process into two stages: 1) initial training, $\alpha_1=0.5, \alpha_4=0.3, \alpha_2=\alpha_3=0$; 2) finetuning, $\alpha_1=\alpha_2=\alpha_3=0.5, \alpha_4=0.3$. 

\paragraph{Loss Function}
$\mathcal{L}_\text{CFM}$ calculates an aligned RMSD by aligning $\mathbf{x_1}$ and $\hat{\mathbf{x_1}}$ to $\mathbf{x}$, while the FAPE loss calculates an averaged RMSD by aligning $\mathbf{\hat{x_1}}$ to each residue frame of $\mathbf{x_1}$, which could be extended to the token frame (Appendix~\ref{appx:token}). Let $\text{align}_{x,i}(y)$ denote aligning $y$ to the $i$-th token frame of $x$, we have
\begin{equation*}
\begin{aligned}
    \mathcal{L}_\text{CFM} &= \mathbb{E}_{t, p_{\text{data}}(x_1), p_t(x|x_1)}\left\|\dfrac{1}{1-t}(\text{align}_{x}(\hat{x_1}(x,t;\theta))-\text{align}_x(x_1))\right\|\\
    &\approx \mathbb{E}_{t, p_{\text{data}}(x_1), p_t(x|x_1)}\left\|\dfrac{1}{1-t} \cdot \dfrac{1}{N}\sum_{i=1}^{N}\left(\text{align}_{x,i}(\hat{x_1}(x,t;\theta))-\text{align}_{x,i}(x_1))\right)\right\|\\
    &\approx \mathbb{E}_{t, p_{\text{data}}(x_1), p_t(x|x_1)}\left\|\dfrac{1}{1-t}\cdot \dfrac{1}{N}\sum_{i=1}^{N}\left(\text{align}_{x_1,i}(\hat{x_1}(x,t;\theta))-\text{align}_{x_1,i}(x_1)\right)\right\|\\
     &\approx \mathbb{E}_{t, p_{\text{data}}(x_1), p_t(x|x_1)}\left\|\dfrac{1}{1-t}\cdot \dfrac{1}{N}\sum_{i=1}^{N}\left(\text{align}_{x_1,i}(\hat{x_1}(x,t;\theta))-x_1\right)\right\|\\
     &=\mathcal{L}_\text{CFM-FAPE}
\end{aligned}
\end{equation*}

\begin{proposition}
    $\text{align}_\mathbf{x_1}(\hat{\mathbf{x_1}})=\mathbf{x_1} \iff \mathcal{L}_\text{CFM}=0\iff \mathcal{L}_\text{CFM-FAPE}=0$. 
\end{proposition}
\begin{proof}
When $\text{align}_{\mathbf{x_1}}(\hat{\mathbf{x_1}}) = \mathbf{x_1}$, we have $\forall i, \text{align}_{\mathbf{x_1}, i}(\hat{\mathbf{x_1}}) = \mathbf{x_1}$. As a result, $\mathcal{L}_\text{CFM} = \mathcal{L}_\text{CFM-FAPE} = 0$. This establishes that:
\begin{equation}
\label{equation:appx:loss_1}
\text{align}_{\mathbf{x_1}}(\hat{\mathbf{x_1}}) = \mathbf{x_1} \iff \mathcal{L}_\text{CFM} = 0 \quad \text{and} \quad \text{align}_{\mathbf{x_1}}(\hat{\mathbf{x_1}}) = \mathbf{x_1} \iff \mathcal{L}_\text{CFM-FAPE} = 0.
\end{equation}
Now, assume $\mathcal{L}_\text{CFM} = 0$. Suppose $\text{align}_{\mathbf{x_1}}(\hat{\mathbf{x_1}}) \neq \mathbf{x_1}$. Then for all transformations $R$ and $t$, we have $R\hat{\mathbf{x_1}} + t \neq \mathbf{x_1}$, which implies: \(
\|\text{align}_{\mathbf{x_1}}(\hat{\mathbf{x_1}}) - \mathbf{x_1}\| \neq 0\), leading to $\mathcal{L}_\text{CFM} \neq 0$. This is a contradiction. Therefore, $\text{align}_{\mathbf{x_1}}(\hat{\mathbf{x_1}}) = \mathbf{x_1}$. This proves that
\begin{equation}
\label{equation:appx:loss_2}
\mathcal{L}_\text{CFM} = 0 \iff \text{align}_{\mathbf{x_1}}(\hat{\mathbf{x_1}}) = \mathbf{x_1}.
\end{equation}
Similarly, assume $\mathcal{L}_\text{CFM-FAPE} = 0$. Suppose $\text{align}_{\mathbf{x_1}}(\hat{\mathbf{x_1}}) \neq \mathbf{x_1}$. Then:
\(
\|\text{align}_{\mathbf{x_1}, i}(\hat{\mathbf{x_1}}) - \mathbf{x_1}\| \neq 0,
\)
which leads to $\mathcal{L}_\text{CFM-FAPE} \neq 0$, again a contradiction. Therefore, $\text{align}_{\mathbf{x_1}}(\hat{\mathbf{x_1}}) = \mathbf{x_1}$. This proves that: 
\begin{equation}
\label{equation:appx:loss_3}
\mathcal{L}_\text{CFM-FAPE} = 0 \iff \text{align}_{\mathbf{x_1}}(\hat{\mathbf{x_1}}) = \mathbf{x_1}.
\end{equation}
The proposition is proved by combining equation \ref{equation:appx:loss_1},\ref{equation:appx:loss_2},\ref{equation:appx:loss_3}.
\end{proof}
This means that both $\mathcal{L}_\text{CFM}$ and $\mathcal{L}_\text{CFM-FAPE}$ provides an optimization direction towards minimizing the SE(3) invariant structural difference between the predicted structure and the ground truth structure.
Thus, we adopt $\mathcal{L}_\text{CFM-FAPE}$ as a realistic approximation of  $\mathcal{L}_\text{CFM}$ and adopt it as the training objection during evaluation.

\subsection{Evaluation Details}
\label{appx:eval}

\paragraph{Specifications}
Following \rfdiffusionaa, we use FAD, SAM, IAI, OQO as the selected evaluation set. FAD and SAM is witnessed by both models as training data, while IAI and OQO is not and demonstrates the generalization ability. To further investigate the performance of our method, we conduct experiments on an extended set of 20 ligands (ligands from PDB id 6cjs, 6e4c, 6gj6, 5zk7, 6qto, 6i78, 6ggd, 6cjj, 6i67, 6iby, 6nw3, 6o5g, 6hlb, 6efk, 6gga, 6mhd, 6i8m, 6s56, 6tel, and 6ffe). The extended dataset includes ligand sizes (including hydrogen) ranging from 21 to 104 in length.

\paragraph{Extended Results}
We illustrate the designability (scRMSD) and binding affinity (Vina energy) of \method-N in Figure \ref{appx:fig:extended-result}. The extended evaluation shows that the performance of \method on the extended set is similar to the evaluation set shown in the main article, and demonstrates that \method is able to tackle with almost all kinds of ligands.

\begin{figure}[h]
\includegraphics[width=\textwidth]{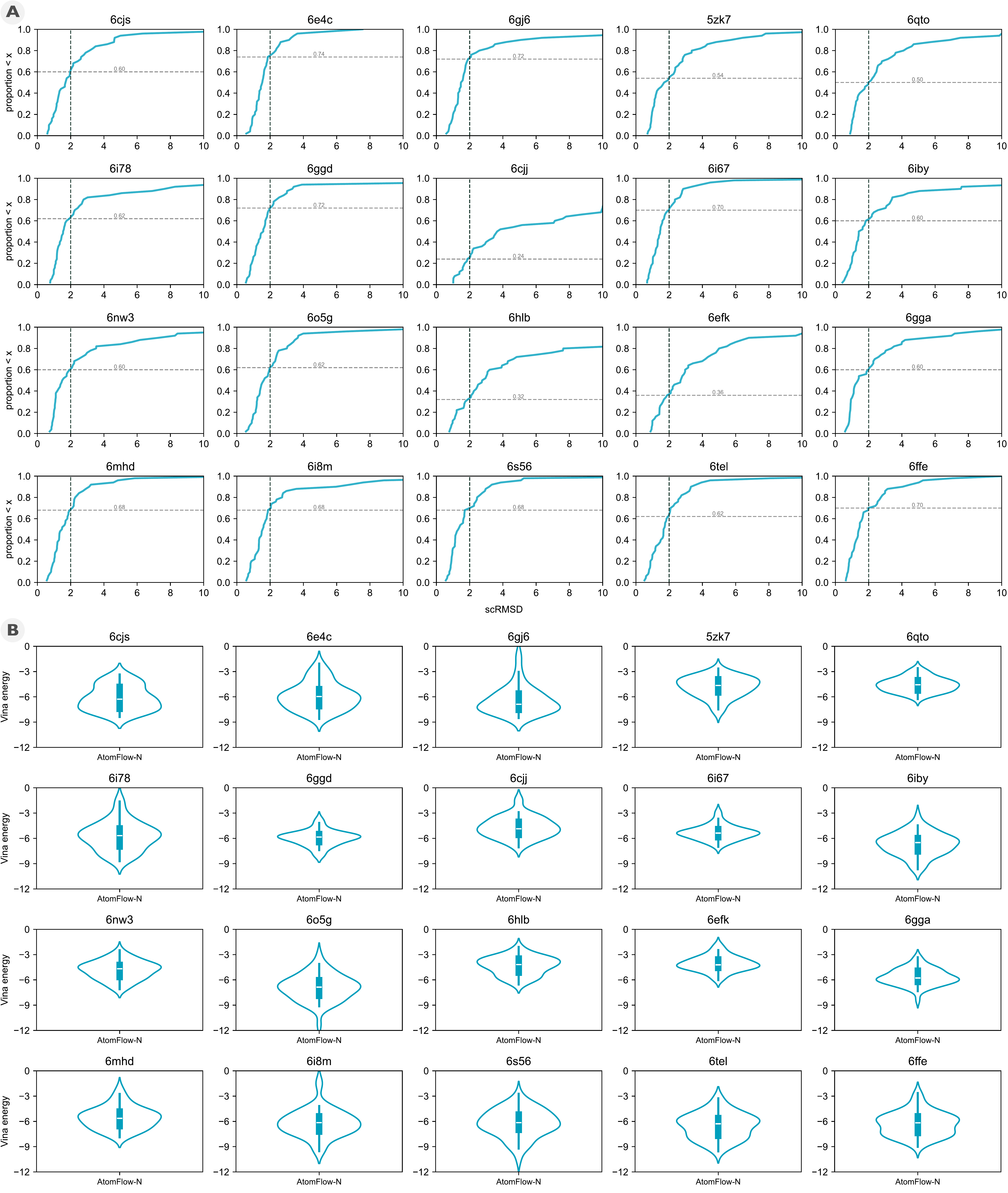}
\caption{A: scRMSD of designs for each ligand in the extended set; B: Vina energy of designs for each ligand in the extended set.}
\label{appx:fig:extended-result}
\end{figure}

\end{document}